\documentclass[journal]{IEEEtrans}
\usepackage{amsmath,amssymb,amsfonts,bm}
\usepackage{algorithm, algorithmic}
\usepackage{breqn}

\usepackage{array}
\usepackage{textcomp}
\usepackage{stfloats}
\usepackage{url}
\usepackage{accents}
\usepackage{color}
\usepackage{verbatim}
\usepackage{graphicx}
\usepackage{cite}
\usepackage{enumerate}
\usepackage{ntheorem}

\usepackage{booktabs}
\usepackage{siunitx}

\usepackage{graphicx}
\usepackage{subcaption}

\theoremstyle{plain}
\theoremseparator{.}  

\makeatletter
\def\widebar{\accentset{{\cc@style\underline{\mskip8mu}}}}
\makeatother
\newtheorem{theorem}{Theorem}
\newtheorem{lemma}{Lemma}
\newtheorem{proposition}{Proposition}
\newtheorem{corollary}{Corollary}
\newtheorem{remark}{Remark}

\newenvironment{proof}{{\it Proof:}}{\hfill $\blacksquare$\par}

\newcolumntype{P}[1]{>{\centering\arraybackslash}p{#1}} 
\newcommand{\tr}{\operatorname{tr}}

\begin{document}

\title{Optimal Low-Dimensional Structures of ISAC Beamforming: Theory and Efficient Algorithms} 

\author{Xiaotong Zhao, Mian Li, Ya-Feng Liu, Qingjiang Shi, and Anthony Man-Cho So
        
\thanks{Xiaotong Zhao and Anthony Man-Cho So are with the Department of Systems Engineering and Engineering Management, The Chinese University of Hong Kong,
 Hong Kong, SAR, China (e-mail: xiaotongzhao@cuhk.edu.hk; manchoso@se.cuhk.edu.hk).}
\thanks{Mian Li is with the School of Science and Engineering, The Chinese University of Hong Kong, Shenzhen, and Shenzhen Research Institute of Big Data, Shenzhen 518172, China (e-mail: mianli1@link.cuhk.edu.cn).}
\thanks{ Ya-Feng Liu with the Ministry of Education Key Laboratory of Mathematics and Information Networks, School of Mathematical Sciences, Beijing University of Posts and Telecommunications, Beijing 102206, China (email: yafengliu@bupt.edu.cn)}
\thanks{Qingjiang Shi is with the School of Computer Science and Technology, Tongji University, Shanghai 201804, China, and also with the Shenzhen Research Institute of Big Data, Shenzhen 518172, China (e-mail: shiqj@tongji.edu.cn).}
}


\maketitle

\begin{abstract}
 Transmit beamforming design is a fundamental problem in integrated sensing and communication (ISAC) systems. Numerous methods have been proposed to jointly optimize key performance metrics such as the signal-to-interference-plus-noise ratio and Cramér-Rao bound. However, the computational complexity of these methods often grows rapidly with the number of transmit antennas at the base station (BS). To tackle this challenge, we prove a fundamental structural property of the ISAC beamforming problem, i.e., there exists an optimal solution exhibiting a low-dimensional structure. This leads to an equivalent reformulation of the problem with dimension related to the number of users rather than the number of BS antennas, thereby enabling the development of low-complexity algorithms. When applying the interior-point method to the reformulated problem, we achieve up to six orders of magnitude in complexity reduction when the number of antennas exceeds the number of users by an order of magnitude. To further reduce the complexity, we develop a balanced augmented Lagrangian method to solve the reformulated problem. The proposed algorithm maintains optimality while achieving a computational complexity that scales quartically with the number of users. Our simulation results demonstrate that the proposed R-BAL method can achieve a speedup of more than $10000\times$ over the conventional IPM in massive MIMO scenarios.

\end{abstract}

\section{Introduction}
\subsection{Background}
 \IEEEPARstart{T}{he} sixth-generation (6G) wireless networks are envisioned to enable transformative applications, such as autonomous driving, smart manufacturing, extended reality, and digital twin, where high-precision sensing and ultra-reliable communications must coexist \cite{chafii2023twelve,dong2025communication}. To meet this dual requirement, integrated sensing and communication (ISAC) has emerged as a pivotal technology, gaining widespread recognition in both academia and industry \cite{liu2020joint,liu2023seventy}. Recently, ISAC has been recognized as one of the six critical usage scenarios in the recommendation for IMT-2030 \cite{wp5d2023draft}.

Unlike traditional systems that operate communication and radar detection independently, ISAC integrates both functionalities into a unified framework, enabling shared spectrum utilization, hardware platforms, and joint signal processing \cite{zhang2021overview}. This integration requires advanced beamforming techniques capable of simultaneously optimizing sensing and communication performance \cite{liu2020joint,liu2024survey}. Existing methodologies can be broadly categorized into three groups: Radar-centric \cite{mealey2007method}, communication-centric \cite{david2015cellular}, and joint design approaches \cite{hassanien2015dual}. Among these, joint beamforming design provides greater flexibility in beamforming and waveform optimization, thus allowing for effective tradeoffs between communication and sensing performance. This advantage has stimulated growing interest in joint transmit beamforming design \cite{liu2022integrated,liu2020joint}.

In joint transmit beamforming design for ISAC systems, various key performance indicators can be used to evaluate system capabilities. On the communication side, performance is typically quantified using metrics such as the achievable sum rate \cite{wu2024efficient} and signal-to-interference-plus-noise ratio (SINR) \cite{liu2020joint}. For sensing, various criteria have been adopted, including beampattern matching \cite{liu2018mu,nguyen2023multiuser}, signal-to-clutter-plus-noise ratio (SCNR) \cite{cui2013mimo}, and the Cramér-Rao bound (CRB) for target parameter estimation \cite{liu2021cramer}. The CRB is particularly noteworthy among these, as it provides a theoretical lower bound on the estimation error, making it a fundamental benchmark for sensing accuracy \cite{xiong2023fundamental}.

\subsection{Related Work}
 Several approaches have been developed for CRB-based ISAC beamforming optimization. The seminal work \cite{liu2021cramer} considered CRB minimization under SINR and transmit power constraints and showed that certain semidefinite relaxation (SDR; see~\cite{luo2010semidefinite} for an overview), whose solution can be obtained using interior-point methods (IPMs), remains tight for this problem. Building on this, the work \cite{wu2024efficient} introduced a weighted sum optimization of communication rate and CRB, achieving global optimality through a combined SDR and branch-and-bound (BB) approach. Although offering superior performance, the BB method has a complexity that grows exponentially with the number of users. Recently, the work \cite{zou2024energy} developed an energy-efficient beamforming design using successive convex approximation (SCA) with subproblems that are solvable by IPMs. In \cite{zhu2023information}, an integrated weighted minimum mean square error (WMMSE) and SDR method was proposed. However, these approaches still rely on computationally intensive interior-point or BB methods and thus are challenging to implement. Although the work \cite{wu2025new} made advances with a first-order algorithm that reduces complexity to the cubic order in the number of base station (BS) antennas, it remains impractical for large-scale systems.

A promising approach to reducing computational complexity lies in leveraging the inherent low-dimensional structures of beamforming optimization problems \cite{bjornson2014optimal,zhao2022rethinking,zhao2025universal,fang2025optimal}. Prior research has demonstrated substantial benefits of such an approach. Indeed, it (1) reveals fundamental properties of optimal beamforming directions \cite{bjornson2014optimal}; (2) provides rigorous performance guarantees in extreme operating regimes \cite{bjornson2014optimal}; (3) enables the development of efficient algorithms, including low-complexity beamforming designs \cite{zhao2022rethinking} and distributed coordination schemes with minimal interaction overhead \cite{zhao2023communication}. Furthermore, the structural insights obtained from the said approach have proven valuable for developing model-driven deep learning approaches, particularly within deep unfolding architectures \cite{xia2019deep}.

\subsection{Motivation and Contributions}
The recent work \cite{zhao2025universal} showed that diverse beamforming scenarios exhibit a universal low-dimensional subspace structure. These scenarios encompass secure communications, multi-cell coordination, and power minimization. Notably, this framework is also effective in certain ISAC beamforming configurations, such as those studied in \cite{liu2018mu} and \cite{fang2025optimal}. The core theoretical insight reveals that the optimal beamforming matrix typically resides in the row space of the matrix induced by the channel.
However, important exceptions exist, particularly in extended target scenarios \cite{liu2021cramer}. This naturally raises a fundamental question: Does CRB-based joint beamforming for extended targets admit an alternative low-dimensional structure?  Our work not only answers this question affirmatively but also develops a generalizable framework for uncovering latent structures in other challenging beamforming problems.

The main contribution of this paper is twofold.
\begin{itemize}
\item[1)] {\textbf{Discovery of an optimal low-dimensional structure:}}
We rigorously derive an optimal low-dimensional structure for CRB minimization in ISAC systems. This theoretical breakthrough enables us to reformulate the original high-dimensional optimization problem into an equivalent one with drastically reduced dimensionality. Specifically, the complexity of standard IPMs drops from $\mathcal{O}(N_t^{6.5}K^{3.5})$ for the original problem to $\mathcal{O}(K^{10})$ for the reduced formulation, where $N_t$ and $K$ represent the number of BS antennas and users, respectively. Our proposed approach is particularly advantageous in massive MIMO scenarios where $N_t \gg K$, offering orders-of-magnitude complexity reduction.
 \item[2)] {\textbf{Low-complexity algorithm design:}}
 We develop a novel balanced augmented Lagrangian (BAL) method that exploits the newly discovered low-dimensional structure to achieve further complexity reduction. Specifically, we leverage the block and rank-one structures of the reduced problem and derive closed-form solutions for the subproblems. The resulting BAL method has a time complexity of $\mathcal{O}(K^4)$, which is a remarkable improvement over existing methods while preserving optimality guarantees.
\end{itemize} 

\emph{Organization:} The rest of the paper is organized as follows. Section II introduces the system model and problem formulation. Section III derives the optimal low-dimensional beamforming structure. Section IV proposes a low-complexity BAL-based algorithm. Section V presents comprehensive numerical results. Finally, Section VI concludes the paper.

\emph{Notation:} Throughout this paper, scalars are denoted by both lower and upper case letters, while vectors and matrices are denoted by boldface lower case and upper case letters, respectively. The space of $M\times N$ complex matrices is denoted by $\mathbb{C}^{M\times N}$. The inverse and trace of a square matrix $\mathbf{A}$ are denoted by $\mathbf{A}^{-1}$ and $\tr(\mathbf{A})$, respectively. The transpose, conjugate transpose, pseudo-inverse, range space, and null space of an arbitrary $\mathbf{A}$ are denoted by $\mathbf{A}^{T}$, $\mathbf{A}^{H}$,  $\mathbf{A}^{\dagger}$, $\mathcal{R}(\mathbf{A})$, and $\mathcal{N}(\mathbf{A})$, respectively.
The Euclidean norm of a vector $\mathbf{a}$ is defined as $\|\mathbf{a}\|=\sqrt{\mathbf{a}^{H}\mathbf{a}}$. The Frobenius norm of a matrix $\mathbf{A}$ is defined as $\|\mathbf{A}\|_{F}=\sqrt{\tr\left(\mathbf{A}^{H}\mathbf{A}\right)}$. 
The Hadamard (element-wise) product of two matrices $\mathbf{A}$ and $\mathbf{B}$ of the same dimensions is denoted by $\mathbf{A} \odot \mathbf{B}$.
Let $\mathbf{0}$, $\mathbf{1}$, and $\mathbf{I}$ denote the all-zero, all-one, and identity matrices of appropriate sizes, respectively. The set $\{1,2,\ldots,K\}$ is abbreviated as $[K]$. For a closed convex set $\mathcal{S}$, its indicator function, denoted by $\mathbb{I}_{\mathcal{S}}(\cdot)$, is defined as $\mathbb{I}_{\mathcal{S}}({\mathbf x}) = 0$ if ${\mathbf x} \in \mathcal{S}$ and $+\infty$ otherwise. Given a vector ${\mathbf a}$, the diagonal matrix formed by putting the elements of ${\mathbf a}$ on its main diagonal is denoted by $\mbox{Diag}({\mathbf a})$. Given a matrix ${\mathbf A}$, the column vector formed by the diagonal elements of ${\mathbf A}$ is denoted by $\mbox{diag}({\mathbf A})$.

\section{System Model and Problem Formulation}
\subsection{System Model}
We consider a downlink massive MIMO ISAC system as in \cite{liu2021cramer}, where a BS equipped with $N_t$ transmit antennas and $N_r$ receive antennas simultaneously serves $K$ single-antenna users while detecting an extended target. Here, the extended target is generally modeled as a surface comprising numerous distributed point-like scatterers, as exemplified by a vehicle or a pedestrian in motion on a roadway.  To avoid information loss in target sensing, the system typically requires $N_t <N_r$.

 Let $\mathbf{\Phi}\in\mathbb{C}^{N_t\times L}$ denote the transmitted baseband signal matrix, where $L>
 N_t$ is the length of the radar pulse (or communication frame).
 The signal matrix $\mathbf{\Phi}$ is formed by combining linearly precoded radar waveforms and communication symbols via
\begin{equation}\notag
    \mathbf{\Phi}=\sum_{k=1}^K\mathbf{w}_k\mathbf{s}_k^H+\mathbf{W}_A\mathbf{S}_A^H,
\end{equation}
where $\mathbf{s}_k\in\mathbb{C}^{L\times 1}$ is the data symbol for the $k$-th communication user and $\mathbf{S}_A\in\mathbb{C}^{L\times N_t}$ is the sensing signal matrix, which are precoded by the communication beamforming vector $\mathbf{w}_k\in\mathbb{C}^{N_t\times 1}$ and the auxiliary beamforming matrix $\mathbf{W}_A\in\mathbb{C}^{N_t\times N_t}$, respectively. We adopt the asymptotic orthogonality assumption from \cite{liu2021cramer}, which states that for sufficiently large $L$, the rows of the data stream matrix $\tilde{\mathbf{S}}= [\mathbf{s}_1,\mathbf{s}_2,\ldots,\mathbf{s}_K,\mathbf{S}_A]^H$ are approximately orthogonal, i.e., $\frac{1}{L}\tilde{\mathbf{S}}\tilde{\mathbf{S}}^H\approx \mathbf{I}_{K+N_t}$.
\begin{remark}
    In massive MIMO systems, the number of BS transmit antennas typically far exceeds the number of users \cite{bjornson2019massive}, i.e., $N_t\gg K$. Therefore, for BSs equipped with extremely large antenna arrays, an efficient beamforming algorithm must achieve computational complexity that scales linearly with $N_t$, or, preferably, remains independent of $N_t$ at each iteration.
\end{remark}

\subsection{Problem Formulation}
\subsubsection{Communication metric} The received signal $\mathbf{y}_k $ at the $k$-th user is given by
\begin{equation}\notag
    \mathbf{y}_k = \mathbf{h}_k^H\mathbf{\Phi}+\mathbf{n}_C,
\end{equation}
where $\mathbf{h}_k\in\mathbb{C}^{N_t\times 1}$ denotes the communication channel vector between the BS and the $k$-th user, which is assumed to be known at the BS; $\mathbf{n}_C$ represents an additive white Gaussian noise (AWGN) with each entry having mean zero and variance $\sigma_C^2$. The SINR for the $k$-th user can be expressed as
\begin{equation}\notag
    \tilde{\gamma}_{k} = \frac{|\mathbf{h}_{k}^{H}\mathbf{w}_{k}|^{2}}{\sum_{i=1,i\neq k}^{K}|\mathbf{h}_{k}^{H}\mathbf{w}_{i}|^{2} + \| \mathbf{h}_{k}^{H}\mathbf{W}_{A} \|^{2} + \sigma_{C}^{2}}.
\end{equation}
We use SINR as the communication metric, which is constrained to exceed a predefined threshold $\Gamma_k$.

\subsubsection{ Sensing metric}
 For sensing, the BS transmits $\mathbf{\Phi}$ to sense the target. The reflected echo signal received at the BS is given by
 \begin{equation}\notag
     \mathbf{Y}_R=\mathbf{G}\mathbf{\Phi} +\mathbf{N}_R,
 \end{equation}
where $ \mathbf{N}_R \in \mathbb{C}^{N_r \times L} $ is an AWGN matrix with each entry having mean zero and variance $\sigma_R^2 $, and $ \mathbf{G} \in \mathbb{C}^{N_r \times N_t} $ denotes the target response matrix. In this paper, we adopt an extended target model following \cite{liu2021cramer}, where $\mathbf{G}$ takes the form
\begin{equation}\notag
    \mathbf{G} = \sum_{m=1}^{N_s} \alpha_m \mathbf{b}(\theta_m) \mathbf{a}^H (\theta_m).
\end{equation}
Here, $N_s$ is the number of scatterers; $ \alpha_m $ and $ \theta_m $ represent the reflection coefficient and azimuth angle of the $m $-th target, respectively; $ \mathbf{a}(\theta_m) $ and $ \mathbf{b}(\theta_m) $ are steering vectors of the transmit and receive antennas, respectively.

In practice, the number of scatterers $ N_s $ is often unknown a priori. Therefore, we focus on estimating the complete matrix $ \mathbf{G} $. The CRB for this estimation is given by \cite{liu2021cramer}
\begin{equation}\notag
    \operatorname{CRB}(\mathbf{G}) = \frac{\sigma_s^2 N_r}{L} \operatorname{Tr}(\mathbf{R}_{\Phi}^{-1}),
\end{equation}
where 
\begin{equation}\notag
    \mathbf{R}_{\Phi} = \frac{1}{L} \mathbf{\Phi} \mathbf{\Phi}^H = \sum_{k =1}^K \mathbf{w}_k \mathbf{w}_k^H + \mathbf{W}_A \mathbf{W}_A^H
\end{equation}
is the sample covariance matrix of $ \mathbf{\Phi}$ due to the orthogonal data stream assumption.

Based on the above discussion, the beamforming optimization problem in the extended target scenario can be expressed as \cite{liu2021cramer}
\begin{equation}\label{eq_extended_mu_isac_problem}
\begin{aligned}
 \min _{\mathbf{W}_{DF}} &\operatorname{tr}\left(\left(\mathbf{W}_{DF}\mathbf{W}_{DF}^H\right)^{-1}\right) \\
\text {s.t. } &\operatorname{tr}\left(\mathbf{Q}_{k}\mathbf{w}_{k}\mathbf{w}_{k}^H\right)-\sum_{j\neq k}\Gamma_k\operatorname{tr}\left(\mathbf{Q}_{k}\mathbf{w}_{j}\mathbf{w}_{j}^H\right)\\
&~~~~~~~-\Gamma_k\operatorname{tr}\left(\mathbf{Q}_{k}\mathbf{W}_{A}\mathbf{W}_{A}^H\right)\geq \Gamma_k \sigma^2_C,\forall k,\\
& \left\|\mathbf{W}_{D F}\right\|_F^2 \leq P_T,
\end{aligned}
\end{equation}
where $\mathbf{W}_{DF}=\left[\mathbf{w}_{1},\mathbf{w}_{2},\ldots,\mathbf{w}_{K},\mathbf{W}_{A}\right]$ and $\mathbf{Q}_{k}=\mathbf{h}_{k}\mathbf{h}_{k}^H, \forall k $. 
Using the results in \cite{liu2021cramer}, one can show that the SDR of problem \eqref{eq_extended_mu_isac_problem} is tight. This establishes the equivalence between the original problem and its SDR
\begin{equation}\label{eq_extended_mu_isac_problem_sdr}
\begin{aligned}
 \min_{\{\mathbf{W}_{k}\}_{k=1}^{K+1}} &\operatorname{tr}\left(\left(\sum_{k=1}^{K+1}\mathbf{W}_k\right)^{-1}\right) \\
\text {s.t. } ~~~&\rho_k\operatorname{tr}\left(\mathbf{Q}_{k}\mathbf{W}_{k}\right)-\sum_{j=1}^{K+1}\operatorname{tr}\left(\mathbf{Q}_{k}\mathbf{W}_{j}\right)\geq\sigma^2_C,\forall k,\\
& \sum_{k=1}^{K+1}\operatorname{tr}\left(\mathbf{W}_{k}\right) \leq P_T,\mathbf{W}_{k}\succeq\mathbf{0},\forall k\in [K+1],
\end{aligned}
\end{equation}
where $\rho_k= 1+\Gamma_k^{-1}$. Unless otherwise specified, we use ``$\forall k$'' as a shorthand for ``$\forall k\in [K]$''.

Note that problem \eqref{eq_extended_mu_isac_problem_sdr} is convex, and the work \cite{liu2021cramer} proposed to solve it using an IPM from standard toolboxes such as CVX. However, such an approach has a computational complexity of $\mathcal{O}(N_t^{6.5}K^{3.5})$, which is prohibitive in systems with large antenna arrays. To address this practical limitation, our work seeks to exploit the inherent low-dimensional beamforming structure of the optimal solution to problem \eqref{eq_extended_mu_isac_problem_sdr}, thereby enabling highly efficient implementations without compromising optimality.

\begin{remark}
    In \cite{liu2021cramer}, joint beamforming designs for CRB minimization were proposed for both point and extended target scenarios, with algorithms operating in a space with the large dimension $N_t$. For the point target case, we make a novel theoretical observation: The design problem actually admits a low-dimensional subspace structure when viewed through the framework of \cite{zhao2025universal}. This insight allows the beamforming design dimension to be reduced from $N_t\times K$ to $(K+2) \times K$. However, the extended target case involves a distinct CRB formulation that falls outside the general framework of \cite{zhao2025universal}. Consequently, whether the corresponding design problem admits a low-dimensional structure  remains an open question. Addressing this question is one of the primary objectives of this paper.
\end{remark}

\section{Optimal Low-Dimensional Beamforming Structure}
In this section, we first examine the feasibility of problem \eqref{eq_extended_mu_isac_problem_sdr}. Then, we characterize the low-dimensional structure possessed by problem \eqref{eq_extended_mu_isac_problem_sdr} and show how to exploit this structure to efficiently compute a rank-one optimal solution to problem \eqref{eq_extended_mu_isac_problem_sdr} via IPM.

\subsection{Feasibility Analysis}
Problem \eqref{eq_extended_mu_isac_problem_sdr} is not always feasible. Intuitively, when the transmit power budget $P_T$ is too small, it is not sufficient to satisfy the SINR constraints. We formally characterize this infeasibility condition through the following proposition.
\begin{proposition}\label{pro_infea_condition}
    Problem \eqref{eq_extended_mu_isac_problem_sdr} (or equivalently, problem \eqref{eq_extended_mu_isac_problem}) is infeasible if and only if $P_T< P_{\operatorname{low}}$, where $P_{\operatorname{low}}$ is given by
\begin{equation}\label{eq_p_low}
    P_{\operatorname{low}}=\sum_{k=1}^K\lambda_k
\end{equation}
with each $\lambda_k$ computable via the fixed-point system 
\begin{equation}\label{eq_lambda_fixedpoint}
    \lambda_k = \frac{\sigma_C^2}{\left(1+\frac{1}{\gamma_k}\right)\widebar{\mathbf{h}}_k^H    \left(\mathbf{H}^H\mathbf{H}+\sum_{i=1}^K \frac{\lambda_i}{\sigma_C^2} \widebar{\mathbf{h}}_i \widebar{\mathbf{h}}_i^H\right)^{-1} \widebar{\mathbf{h}}_k}, ~\forall k.
\end{equation}
Here, $\mathbf{H}=[\mathbf{h}_{1},\mathbf{h}_{2},\ldots,\mathbf{h}_{K}]$ and $\widebar{\mathbf{h}}_i= \mathbf{H}^H\mathbf{h}_i\in\mathbb{C}^{K\times 1},\forall i$.
\end{proposition}
\begin{proof}
    See Appendix \ref{app_pro_infea_condition}.
\end{proof}

The fixed-point system in \eqref{eq_lambda_fixedpoint} requires
the inversion of $K\times K$ matrices. Each inversion has a time complexity of $\mathcal{O}(K^3)$, which represents a significant reduction from the conventional $\mathcal{O}(N_t^3)$ complexity required by standard methods \cite{schubert2004solution}. Such efficiency is made possible through an exploitation of the low-dimensional subspace structure identified in \cite{zhao2025universal}; see Appendix \ref{app_pro_infea_condition} for details.

In the rest of this paper, we assume that problem \eqref{eq_extended_mu_isac_problem_sdr} is feasible.

\subsection{Optimal Solution Structure Analysis}
Through a careful analysis of the Karush-Kuhn-Tucker (KKT) conditions associated with problem \eqref{eq_extended_mu_isac_problem_sdr}, we establish the following fundamental structural properties of its optimal solutions.
\begin{theorem}\label{theorem_low_dimen}
   There exists an optimal solution $\{\mathbf{W}_{k}^{\star}\}_{k=1}^{K+1}$ to  problem \eqref{eq_extended_mu_isac_problem_sdr} that  satisfies the following properties:

   \begin{enumerate}[(a)]
  \item $\mathcal{R}(\mathbf{W}_{k}^{\star})\subseteq \mathcal{R}(\mathbf{H}),\forall k$;
  \item $\mathcal{R}(\mathbf{W}_{K+1}^{\star})\subseteq \mathcal{N}(\mathbf{H}^H)$;~and
  \item $\mathbf{W}_{K+1}^{\star}=\theta\mathbf{U}_C\mathbf{U}_C^H$, where the columns of $\mathbf{U}_C\in\mathbb{C}^{N_t\times (N_t-K)}$ form an orthonormal basis of  $\mathcal{N}(\mathbf{H}^H)$ and $\theta = \frac{P_T-\sum_{k=1}^{K}\operatorname{tr}\left(\mathbf{W}_{k}^{\star}\right)}{N_t-K}$.
\end{enumerate}
\end{theorem}
\begin{proof}
    See Appendix \ref{app_theorem_low_dimen}.
\end{proof}

The structural properties established in Theorem \ref{theorem_low_dimen} offer a way to substantially simplify problem \eqref{eq_extended_mu_isac_problem_sdr}. First, property (c) reveals that ${\mathbf W}_{K+1}^{\star}$ can be directly computed in closed form once $\{\mathbf{W}_{k}^{\star}\}_{k=1}^K$ is determined. Second, property (b) implies the key orthogonality condition $\mathbf{Q}_k\mathbf{W}_{K+1}^{\star}=\mathbf{0} $, $\forall k$. In particular, we see that ${\mathbf W}_{K+1}^{\star}$ does not have any interference effect on the communication SINR constraints. These insights allow us to equivalently transform problem \eqref{eq_extended_mu_isac_problem_sdr} into the following problem involving only $\{\mathbf{W}_{k}\}_{k=1}^K$:
\begin{equation}\label{eq_extended_mu_isac_problem_wc}
\begin{aligned}
 \min _{\{\mathbf{W}_k\}_{k=1}^K} &\text{tr}\left(\left(\sum_{k=1}^{K}\mathbf{W}_k\right)^{\dagger}\right) +\frac{(N_t-K)^2}{P_T-\operatorname{tr}\left(\sum_{k=1}^{K}\mathbf{W}_k\right)} \\
\text {s.t. } ~~&\rho_k\operatorname{tr}\left(\mathbf{Q}_{k}\mathbf{W}_{k}\right)-\sum_{j=1}^K\operatorname{tr}\left(\mathbf{Q}_{k}\mathbf{W}_{j}\right)\geq  \sigma^2_C,\forall k,\\
& \sum_{k=1}^{K}\operatorname{tr}\left(\mathbf{W}_{k}\right) \leq P_T,\mathbf{W}_{k}\succeq\mathbf{0},\forall k.
\end{aligned}
\end{equation}

\noindent Third, property (a) reveals that the optimal solution $\{\mathbf{W}_{k}^{\star}\}_{k=1}^K$ admits a low-dimensional subspace structure, i.e., $\mathbf{W}_{k}^{\star}=\mathbf{H}\mathbf{V}_k\mathbf{H}^H $ with $\mathbf{V}_k\in\mathbb{C}^{K\times K} $, $\forall k$. Building on this, we have the following theorem.
 \begin{theorem}
      Problem \eqref{eq_extended_mu_isac_problem_sdr} is equivalent to
\begin{equation}\label{eq_extended_mu_isac_problem_wc_lds_sdr}
\begin{aligned}
 \min _{\{\mathbf{V}_k\}_{k=1}^K} &\operatorname{tr}\left(\widebar{\mathbf{H}}^{-1}\widebar{\mathbf{R}}_V^{-1}\right) +\frac{(N_t-K)^2}{P_T-\operatorname{tr}\left(\widebar{\mathbf{H}}\widebar{\mathbf{R}}_V\right)} \\
\text {\rm s.t. }  ~~&\rho_k\operatorname{tr}\left(\widebar{\mathbf{Q}}_{k}\mathbf{V}_k\right)-\operatorname{tr}\left(\widebar{\mathbf{Q}}_{k}\widebar{\mathbf{R}}_V\right)\geq \sigma^2_C,\forall k,\\
& \operatorname{tr}\left(\widebar{\mathbf{H}}\widebar{\mathbf{R}}_V\right) \leq P_T,~\mathbf{V}_k\succeq\mathbf{0},\forall k,
\end{aligned}
\end{equation}
where $\widebar{\mathbf{R}}_V=\sum_{k=1}^K\mathbf{V}_k$, $\widebar{\mathbf{H}}=\mathbf{H}^H\mathbf{H}\in\mathbb{C}^{K\times K}$, and $\widebar{\mathbf{Q}}_k=\mathbf{H}^H\mathbf{h}_k\mathbf{h}_k^H\mathbf{H}\in\mathbb{C}^{K\times K},\forall k$.
 \end{theorem}
 
\noindent Note that the dimension of the matrix variable has been significantly reduced from $N_t\times N_t$ to $K\times K$, which facilitates low-complexity algorithm design.

While problem \eqref{eq_extended_mu_isac_problem_wc_lds_sdr} involves the matrix $\widebar{\mathbf{H}} = \mathbf{H}^H\mathbf{H}$ in both the objective function and constraints, it does not fundamentally complicate algorithm design, as it is simply a linear operator acting on $\{ {\mathbf V}_k \}_{k=1}^K$. Nevertheless, the formulation can be further simplified by leveraging an equivalent subspace representation derived from the singular value decomposition (SVD) of the channel matrix.

Specifically, let $\mathbf{H} = \tilde{\mathbf{U}}\tilde{\boldsymbol{\Sigma}}\tilde{\mathbf{V}}^H$ be the compact SVD of $\mathbf{H}$, where $\tilde{\mathbf{U}} \in \mathbb{C}^{N_t \times K}$ forms an orthonormal basis of the range space of $\mathbf{H}$. According to Theorem \ref{theorem_low_dimen}(a), the optimal beamforming matrices can alternatively be expressed as $\mathbf{W}_{k}=\tilde{\mathbf{U}}\mathbf{X}_k\tilde{\mathbf{U}}^H $ with $\mathbf{X}_k\in\mathbb{C}^{K\times K},\forall k$. This representation leads to the following equivalent and more compact reformulation of problem \eqref{eq_extended_mu_isac_problem_wc_lds_sdr}.
\begin{corollary}\label{coro_reduced}
Problem \eqref{eq_extended_mu_isac_problem_sdr} is equivalent to
    \begin{equation}\label{eq_extended_mu_isac_problem_wc_lds_sdr_brief}
\begin{aligned}
 \min _{\{\mathbf{X}_k\}_{k=1}^K} &\operatorname{tr}\left(\tilde{\mathbf{R}}_X^{-1}\right) +\frac{(N_t-K)^2}{P_T-\operatorname{tr}\left(\tilde{\mathbf{R}}_X\right)} \\
\text {\rm s.t. } ~~&\rho_k\operatorname{tr}\left(\tilde{\mathbf{Q}}_{k}\mathbf{X}_k\right)-\operatorname{tr}\left(\tilde{\mathbf{Q}}_{k}\tilde{\mathbf{R}}_X\right)\geq \sigma^2_C,\forall k,\\
& \operatorname{tr}\left(\tilde{\mathbf{R}}_X\right) \leq P_T,~\mathbf{X}_k\succeq\mathbf{0},\forall k,
\end{aligned}
\end{equation}
where $\tilde{\mathbf{R}}_X{=}\sum_{k=1}^K\mathbf{X}_k$, and $\tilde{\mathbf{Q}}_k{=}\tilde{\mathbf{U}}^H\mathbf{h}_k\mathbf{h}_k^H\tilde{\mathbf{U}}\in\mathbb{C}^{K\times K},\forall k$.
\end{corollary}

 The formulations in \eqref{eq_extended_mu_isac_problem_wc_lds_sdr} and \eqref{eq_extended_mu_isac_problem_wc_lds_sdr_brief} are mathematically equivalent to each other. Although one needs to perform an SVD of the channel matrix to form problem \eqref{eq_extended_mu_isac_problem_wc_lds_sdr_brief}, this operation leads to a more compact formulation and facilitates clearer and more concise expressions in subsequent derivations. Therefore, in what follows, we develop efficient algorithms based on the reformulated problem \eqref{eq_extended_mu_isac_problem_wc_lds_sdr_brief}.
\begin{remark}
    The low-dimensional structure revealed in this work has implications not only on the current CRB minimization framework but also on broader settings, such as maximization of a weighted combination of the sum rate and the CRB\cite{wu2024efficient}, maximization of the sum rate under CRB constraints\cite{hua2023mimo}, and energy-efficient beamforming design for ISAC\cite{zou2024energy}.
\end{remark}

\subsection{Rank-One Optimal Solution via IPM}
 Problem \eqref{eq_extended_mu_isac_problem_wc_lds_sdr_brief} is convex and can be solved by a standard IPM. The proposed algorithm is termed reduced-IPM (R-IPM), whose computational complexity is $\mathcal{O}(K^{10})$. This should be contrasted with the $\mathcal{O}(N_t^{6.5}K^{3.5})$ complexity of standard IPMs for solving the original problem in \cite{liu2021cramer}.\footnote{For simplicity, we omit the total iteration number in the complexity analysis in the rest of this paper. The per-iteration complexity of the IPM can be found in  \cite{wang2014outage}.} Thus, our approach significantly reduces the complexity and is more preferred in massive MIMO systems.

However, an optimal solution to problem \eqref{eq_extended_mu_isac_problem_wc_lds_sdr_brief}  is not guaranteed to be of rank-one, even though it admits a low-dimensional structure. The following theorem provides a constructive method to extract a rank-one optimal solution to the original problem \eqref{eq_extended_mu_isac_problem_sdr} from an optimal solution to problem \eqref{eq_extended_mu_isac_problem_wc_lds_sdr_brief}.

\begin{theorem}\label{thm_extract_rank1}
    Given an optimal solution $\{\mathbf{X}^{\star}_{k}\}_{k=1}^{K}$ to problem \eqref{eq_extended_mu_isac_problem_wc_lds_sdr_brief}, an optimal solution $\{ {\mathbf W}_k^\star \}_{k=1}^{K+1}$ to problem \eqref{eq_extended_mu_isac_problem_sdr} can be constructed via
\begin{equation}\label{thm_extract_rank_1}
\mathbf{W}_{k}^{\star}=\frac{\tilde{\mathbf{U}}\mathbf{X}^{\star}_{k}\tilde{\mathbf{Q}}_{k}\mathbf{X}^{\star}_{k}\tilde{\mathbf{U}}^H}{\operatorname{tr}(\tilde{\mathbf{Q}}_{k}\mathbf{X}^{\star}_{k})},\forall k\end{equation}
and
\begin{equation}\label{thm_extract_rank_2}
    \mathbf{W}_{K+1}^{\star}=\mathbf{R}^{\star}_{W}-\sum_{k=1}^K\mathbf{W}_{k}^{\star},
\end{equation}
where
\begin{equation}\notag
    \mathbf{R}^{\star}_{W}=\sum_{k=1}^{K} \tilde{\mathbf{U}}\mathbf{X}^{\star}_{k}\tilde{\mathbf{U}}^H+\frac{P_T-\sum_{k=1}^K\operatorname{tr}(\mathbf{X}^{\star}_{k})}{N_t-K}\mathbf{U}_C\mathbf{U}_C^H.
\end{equation}
Moreover, we have $\operatorname{rank}(\mathbf{W}_{k}^{\star})=1$ for all $k\in[K]$.
\end{theorem}
\begin{proof}
    See Appendix \ref{app_thm_extract_rank1} in the Supplementary Material.
\end{proof}

According to Theorem \ref{thm_extract_rank1}, by setting $\mathbf{w}_k^{\star}=(\mathbf{h}_k^H\tilde{\mathbf{U}}\mathbf{X}^{\star}_{k}\tilde{\mathbf{U}}^H\mathbf{h}_k)^{-\frac{1}{2}}\tilde{\mathbf{U}}\mathbf{X}^{\star}_{k}\tilde{\mathbf{U}}^H\mathbf{h}_k$ for $ k\in [K]$ and $\mathbf{W}_A^{\star}$ to be, say, a Cholesky factor of $ \mathbf{W}_{K+1}^{\star}$, we obtain an optimal solution to problem \eqref{eq_extended_mu_isac_problem}.

\section{Balanced Augmented Lagrangian-Based Low-Complexity Algorithm}
Although the proposed R-IPM algorithm provides a globally optimal solution to problem \eqref{eq_extended_mu_isac_problem_wc_lds_sdr_brief}, its computational complexity of $\mathcal{O}(K^{10})$ becomes unbearable for large values of $K$. To overcome this limitation, we develop a BAL-based method that efficiently solves problem \eqref{eq_extended_mu_isac_problem_wc_lds_sdr_brief} with significantly lower complexity.
\subsection{Brief Introduction of BAL Method and Challenges}
The BAL method is an enhanced version of the classical augmented Lagrangian method (ALM), designed to efficiently solve optimization problems of the form
\begin{equation}\label{eq_linear_equality_constrained_cvx}
\min _{\mathbf{u} \in \mathbb{C}^n} f(\mathbf{u}) \quad \text { s.t. } \quad \mathbf{D u}=\mathbf{b},
\end{equation}
where $f: \mathbb{C}^n \to \mathbb{R} \cup \{+\infty\}$ is a proper closed convex function; $\mathbf{D} \in \mathbb{C}^{m \times n}$ and $\mathbf{b} \in \mathbb{C}^m$ are given complex-valued matrix and vector, respectively.

The BAL method solves problem \eqref{eq_linear_equality_constrained_cvx} through the following iterative procedure:
\begin{equation}\label{eq_BALM}
   \begin{cases}
\mathbf{u}^{t+1} = \arg\min_{\mathbf{u} \in \mathbb{C}^n} \left\{ f(\mathbf{u}) + \frac{1}{2\tau} \|\mathbf{u} - (\mathbf{u}^t - \tau \mathbf{D}^{H}\boldsymbol{\lambda}^t)\|^2 \right\} , \\
\mathbf{p}^{t+1} = \mathbf{D}\left(2\mathbf{u}^{t+1} - \mathbf{u}^t\right) - \mathbf{b}, \\
\boldsymbol{\lambda}^{t+1} = \boldsymbol{\lambda}^t +\tau^{-1}\left(\mathbf{D}\mathbf{D}^{H} + \delta\mathbf{I}\right)^{-1}\mathbf{p}^{t+1}.
\end{cases} 
\end{equation}
Here, $\tau > 0$ is the primal stepsize and $\delta > 0$ is a preset small regularization parameter (e.g., $10^{-4}$) to ensure the positive definiteness of $\mathbf{D}\mathbf{D}^H+\delta \mathbf{I}$.

Unlike the classical ALM, the BAL method achieves computational balance by redistributing the workload between the primal and dual updates. Specifically, the primal update (i.e., the $\mathbf{u}$-subproblem) in \eqref{eq_BALM} often admits a closed-form solution, in contrast to the typically more challenging $\mathbf{u}$-subproblem in the classical ALM. Moreover, the dual update (i.e., the $\boldsymbol{\lambda}$-subproblem) reduces to solving a linear system. Crucially, in many practical applications (including the case discussed in the subsequent subsection), the matrix $\mathbf{D}$ possesses a special structure (e.g., block-diagonal or sparse) that facilitates an efficient computation of the dual update. This ensures that both subproblems remain computationally tractable.
The BAL method enjoys global convergence and achieves the worst-case  $\mathcal{O}(1/t)$ convergence rate, as established in\cite{he2021balanced,ma2023understanding}. For further theoretical and algorithmic details, we refer readers to \cite{he2021balanced}. 

Several recent works have proposed enhancements to the BAL method, including adaptive stepsize strategies \cite{wu2025new}, acceleration techniques \cite{zhang20251}, and prediction-correction schemes \cite{li2025balanced}. This paper does not aim to improve the BAL framework itself. Instead, we focus on exploiting low-dimensional structures to design efficient algorithms. In particular, we show how the standard BAL method can be applied effectively to solve our low-dimensional problem \eqref{eq_extended_mu_isac_problem_wc_lds_sdr_brief} while maintaining low computational complexity.

 To apply the BAL method to problem \eqref{eq_extended_mu_isac_problem_wc_lds_sdr_brief}, we must address the following three key technical challenges:
\begin{itemize}
    \item[1)] \textbf{Inequality constraints handling:} Problem \eqref{eq_extended_mu_isac_problem_wc_lds_sdr_brief} involves $K$ SINR inequality constraints. Direct application of the BAL method would require solving a linearly constrained quadratic program in each dual update step. This approach is significantly more complex than the equality-constrained case, where the dual update admits a closed-form solution as in \eqref{eq_BALM}.
    \item[2)] \textbf{Complicated primal update:} Although the BAL method simplifies the primal update in the classical ALM, the complicated structure of the objective function in \eqref{eq_extended_mu_isac_problem_wc_lds_sdr_brief} remains computationally challenging.   
    \item[3)] \textbf{Large-scale matrix inversion:} The semidefinite constraints introduce dual variables of dimension $K^3$, thus resulting in large-scale matrix inversions in the dual update. This typically incurs a prohibitive computational complexity of $\mathcal{O}(K^9)$.
\end{itemize}

In the following subsections, we will address these challenges and develop an efficient BAL-based algorithm with a significantly lower per-iteration complexity of $\mathcal{O}(K^4)$.

\subsection{Problem Reformulation and Equality Structure}

To address Challenge 1, we first examine whether there exists an optimal solution to problem \eqref{eq_extended_mu_isac_problem_wc_lds_sdr_brief} such that all inequality constraints are \emph{tight}, i.e., they hold as equalities. For the original problem \eqref{eq_extended_mu_isac_problem_sdr}, it has been shown in \cite{wu2025new} that such an optimal solution does exist. However, due to the structural restriction imposed on $\mathbf{W}_{K+1}$ by Theorem \ref{theorem_low_dimen}(c), it may happen that for every optimal solution to problem \eqref{eq_extended_mu_isac_problem_wc_lds_sdr_brief}, at least one constraint holds as strict inequality. We refer to this scenario as the \emph{degenerate case}. The following theorem furnishes a necessary condition for the degenerate case to occur.

\begin{theorem}\label{thm_equality} 
\begin{enumerate}[(a)]
  \item A necessary condition for the degenerate case to occur in problem \eqref{eq_extended_mu_isac_problem_wc_lds_sdr_brief}  is that the inequality
    \begin{equation}\label{eq_inequality_condition}
    \|\mathbf{h}_l\|^2 \sum_{k=1}^K \frac{P_T\|\mathbf{h}_k\|^2+  \sigma_C^2 N_t}{\rho_k\tr\left(\mathbf{Q}_k\mathbf{Q}_l\right)}< P_T 
    \end{equation}
    holds for all $l\in[K]$.
  \item If condition \eqref{eq_inequality_condition} is satisfied, then for any $l\in [K]$, the matrices $\{ {\mathbf W}_k \}_{k=1}^{K+1}$ given by $\mathbf{W}_k=a_k\mathbf{Q}_l$ for $k\in[K]$ and $\mathbf{W}_{K+1}=\frac{P_T}{N_t}\mathbf{I}-\sum_{k=1}^Ka_k\mathbf{Q}_l$ constitute an optimal solution to problem \eqref{eq_extended_mu_isac_problem_sdr}, where $a_k =  \frac{P_T\|\mathbf{h}_k\|^2+  \sigma_C^2 N_t}{\rho_k N_t\tr\left(\mathbf{Q}_k\mathbf{Q}_l\right)}$.
\end{enumerate}
\end{theorem}
\begin{proof}
   See Appendix \ref{app_thm_equality} in the Supplementary Material.
\end{proof}

Theorem \ref{thm_equality}(b) implies that the necessary condition \eqref{eq_inequality_condition} can only be satisfied under highly restrictive settings. Specifically, it indicates that aligning all the communication beamforming vectors in the same direction as $\mathbf{h}_l$ for any $l\in[K]$ is sufficient to achieve the optimal solution. This can only happen when the transmit power $P_T$ is sufficiently high, the SINR thresholds $\Gamma_k$ are sufficiently low for all $k \in [K]$, and all channel vectors are oriented in nearly the same direction. In practical systems, especially those with large antenna arrays at the BS, this condition occurs with probability zero due to the inherent spatial diversity of wireless channels. Even in the rare case where it does occur, constructing an optimal solution remains straightforward, as outlined in Theorem \ref{thm_equality}(b). Therefore, in what follows, we will assume that problem \eqref{eq_extended_mu_isac_problem_wc_lds_sdr_brief} admits an optimal solution for which all SINR constraints are tight.


Based on Corollary \ref{coro_reduced} and Theorem \ref{thm_equality}, the original problem \eqref{eq_extended_mu_isac_problem_sdr} is equivalent to the following linear equality-constrained convex optimization problem:
 \begin{equation}\label{eq_extended_mu_isac_problem_wc_lds_sdr_brief11}
\begin{aligned}
 \min _{\{\mathbf{X}_k\}_{k=1}^K} &\operatorname{tr}\left(\tilde{\mathbf{R}}_X^{-1}\right) +\frac{(N_t-K)^2}{P_T-\operatorname{tr}\left(\tilde{\mathbf{R}}_X\right)}+\mathbb{I}_{\mathcal{X}}\left(\mathbf{X}\right) \\
\text {s.t. } ~~&\rho_k\operatorname{tr}\left(\tilde{\mathbf{Q}}_{k}\mathbf{X}_k\right)-\operatorname{tr}\left(\tilde{\mathbf{Q}}_{k}\tilde{\mathbf{R}}_X\right)= \sigma^2_C,\forall k,
\end{aligned}
\end{equation}
where $\mathbf{X}= \left[\mathbf{X}_1,\mathbf{X}_2,\ldots,\mathbf{X}_K\right]\in\mathbb{C}^{K\times K^2}$, $\mathcal{X}= \{\mathbf{X}\mid\sum_{k=1}^K\operatorname{tr}\left(\mathbf{X}_k\right) \leq P_T,\mathbf{X}_k\succeq\mathbf{0},\forall k\}$, and $\{ \tilde{\mathbf Q}_k \}_{k=1}^K$, $\tilde{\mathbf{R}}_X$ are defined in \eqref{eq_extended_mu_isac_problem_wc_lds_sdr_brief}.


\subsection{Proposed R-BAL Method}
In this subsection, we employ the BAL method to solve problem \eqref{eq_extended_mu_isac_problem_wc_lds_sdr_brief11}. As noted in Challenge 2, solving the $\mathbf{X}$-subproblem is computationally expensive, as it requires computing the proximal mapping of the objective function of problem \eqref{eq_extended_mu_isac_problem_wc_lds_sdr_brief11}.
To address this challenge, we propose a variable-splitting technique to reformulate the original problem into the following form, which is better suited for efficient algorithmic implementation:
\begin{equation}\label{eq_extended_mu_isac_problem_wc_lds_bal}
\begin{aligned}
 \min _{\mathbf{X},\mathbf{Y},\mathbf{Z}} ~&\operatorname{tr}\left(\mathbf{Y}^{-1}\right) +\frac{(N_t-K)^2}{P_T-\operatorname{tr}\left(\mathbf{Z}\right)}+\mathbb{I}_{\mathcal{X}}\left(\mathbf{X}\right) \\
\text {s.t. } ~~&\rho_k\operatorname{tr}\left(\tilde{\mathbf{Q}}_{k}\mathbf{X}_k\right)-\operatorname{tr}\left(\tilde{\mathbf{Q}}_{k}\mathbf{Y}\right) = ~ \sigma^2_C,\forall k,\\
&\sum_{k=1}^K\mathbf{X}_k-\mathbf{Y}=\mathbf{0},\\
&\mathbf{Y}-\mathbf{Z}=\mathbf{0}.
\end{aligned}
\end{equation}

\noindent Problem \eqref{eq_extended_mu_isac_problem_wc_lds_bal} can be equivalently reformulated as a linear equality-constrained convex optimization problem of the form \eqref{eq_linear_equality_constrained_cvx} by vectorizing the matrices $\mathbf{X}$, $\mathbf{Y}$, and $\mathbf{Z}$. Specifically, let $\mathbf{u}= \left[\mathbf{x}^T,\mathbf{y}^T,\mathbf{z}^T\right]^T\in \mathbb{C}^{K^3+2K^2}$ with 
\begin{equation}\notag
  \mathbf{x}=\left(\begin{array}{c}\operatorname{vec}\left(\mathbf{X}_1\right) \\ \vdots \\ \operatorname{vec}\left(\mathbf{X}_K\right)\end{array}\right) \in \mathbb{C}^{K^3},
\end{equation}
$\mathbf{y}=\operatorname{vec}(\mathbf{Y}) \in \mathbb{C}^{K^2}$, and  $\mathbf{z}=\operatorname{vec}(\mathbf{Z})\in \mathbb{C}^{K^2}$. Then, problem \eqref{eq_extended_mu_isac_problem_wc_lds_bal} takes the form \eqref{eq_linear_equality_constrained_cvx} with the following $f$, $\mathbf{D}$, and $\mathbf{b}$:
\begin{itemize}
\setlength{\leftmargin}{0pt} 
\setlength{\itemindent}{-0.25em}
\item The function $f$ is given by $f(\mathbf{u})=f_1(\mathbf{x})+f_2(\mathbf{y})+f_3(\mathbf{z})$ with $f_1(\mathbf{x})=\mathbb{I}_{\mathcal{X}}\left(\mathbf{X}\right)$, $f_2(\mathbf{y})=\tr(\mathbf{Y}^{-1})$, and $f_3(\mathbf{z})=\frac{(N_t-K)^2}{P_T-\operatorname{tr}\left(\mathbf{Z}\right)}$.

\item The coefficient matrix ${\mathbf{D}{=}[\mathbf{A},\mathbf{B},\mathbf{C}]{\in} \mathbb{C}^{(K{+}2K^2){\times} (K^3{+}2K^2)}}$ with
\begin{equation}\notag
  \mathbf{A}=\left(\begin{array}{c}\mathbf{A}_1 \\ \mathbf{A}_2 \\ \mathbf{A}_3\end{array}\right) \in \mathbb{C}^{(K+2K^2)\times K^3},
\end{equation}
\begin{equation}\notag
  \mathbf{B}=\left(\begin{array}{c}\mathbf{B}_1 \\ - \mathbf{I}_{K^2}\\ \mathbf{I}_{K^2}\end{array}\right) \in \mathbb{C}^{(K+2K^2)\times K^2},
\end{equation}
and 
\begin{equation}\notag
  \mathbf{C}=\left(\begin{array}{c}\mathbf{0} \\  \mathbf{0}_{K^2} \\ - \mathbf{I}_{K^2}\end{array}\right) \in \mathbb{C}^{(K+2K^2)\times K^2},
\end{equation}
where 
\begin{equation}\notag
\begin{aligned}
\mathbf{A}_1\!=\!&\left(\begin{array}{cccc}
\rho_1 \operatorname{vec}(\tilde{\mathbf{Q}}_1)^{H} & \!\!\!\!\mathbf{0} &\!\!\!\! \cdots &\!\!\!\! \mathbf{0} \\
\mathbf{0} & \!\!\!\!\rho_2 \operatorname{vec}(\tilde{\mathbf{Q}}_2)^{H} & \!\!\!\!\cdots & \!\!\!\!\mathbf{0} \\
\vdots &\!\!\!\! \vdots& \!\!\!\!\vdots & \!\!\!\!\vdots \\
\mathbf{0} &\!\!\!\! \mathbf{0} & \!\!\!\!\cdots & \!\!\!\!\rho_K \operatorname{vec}(\tilde{\mathbf{Q}}_K)^{H} 
\end{array}\!\right)\\
&\quad\quad\quad\quad\quad\quad\quad\quad\quad\quad\quad\quad\quad\quad\quad\quad\in \mathbb{C}^{K\times K^3},
\end{aligned}
\end{equation}
$\mathbf{A}_2=\mathbf{1}_{K}^T\otimes \mathbf{I}_{K^2}\in \mathbb{C}^{K^2\times K^3},~\mathbf{A}_3=\mathbf{0}\in \mathbb{C}^{K^2\times K^3}$, and $\mathbf{B}_1=-[\operatorname{vec}(\tilde{\mathbf{Q}}_1),\operatorname{vec}(\tilde{\mathbf{Q}}_2),\ldots,\operatorname{vec}(\tilde{\mathbf{Q}}_K)]^H$.

\item The vector
\begin{equation}\notag
  \mathbf{b}=\left(\begin{array}{c}\sigma_C^2\mathbf{1}_K \\ \mathbf{0}_{K^2} \\ \mathbf{0}_{K^2}\end{array}\right) \in \mathbb{C}^{(K+2K^2)\times 1}.
\end{equation}
\end{itemize}
Now, we are ready to present the BAL-based method. Note that the main computation in \eqref{eq_BALM} lies in the $\mathbf{u}$- and $\boldsymbol{\lambda}$-subproblems. Let us study them in turn.

\subsubsection{$\mathbf{u}$-subproblem}
The $\mathbf{u}$-subproblem is separable in $\mathbf{X}$, $\mathbf{Y}$, and $\mathbf{Z}$. Thus, we need to solve the following three subproblems:
\begin{equation}\label{eq_X_subproblem}
     \min _{\mathbf{X}} ~\mathbb{I}_{\mathcal{X}}\left(\mathbf{X}\right) + \frac{1}{2\tau}\|\mathbf{X}-\widetilde{\mathbf{X}}\|_F^2.
\end{equation}
\begin{equation}\label{eq_Y_subproblem}
     \min _{\mathbf{Y}} ~\operatorname{tr}\left(\mathbf{Y}^{-1}\right) + \frac{1}{2\tau}\|\mathbf{Y}-\widetilde{\mathbf{Y}}\|_F^2.
\end{equation}
\begin{equation}\label{eq_Z_subproblem}
     \min _{\mathbf{Z}} ~\frac{(N_t-K)^2}{P_T-\operatorname{tr}\left(\mathbf{Z}\right)} + \frac{1}{2\tau}\|\mathbf{Z}-\widetilde{\mathbf{Z}}\|_F^2.
\end{equation}
As Propositions \ref{pro_subproblem_X}--\ref{pro_subproblem_Z} show, the above subproblems can be efficiently solved.

\begin{proposition}\label{pro_subproblem_X}
    Let $\widetilde{\mathbf{X}}_k=\mathbf{U}_k\boldsymbol{\Sigma}_k\mathbf{U}_k^H$ denote the eigenvalue decomposition of $\widetilde{\mathbf{X}}_k$ for $k \in [K]$, where $\mathbf{U}_k\in\mathbb{C}^{N\times N}$ is a unitary matrix and $\boldsymbol{\Sigma}_k$ is a real diagonal matrix. Then, the optimal solution to the $\mathbf{X}$-subproblem \eqref{eq_X_subproblem} is given by
    \begin{equation}\label{eq_update_X1}
        \mathbf{X}_k=\mathbf{U}_k\boldsymbol{\Lambda}_k\mathbf{U}_k^H,\,\forall k.
    \end{equation}
    Here, $\boldsymbol{\Lambda}_k$'s are diagonal matrices given by 
    \begin{equation}\label{eq_update_X2}
   [\boldsymbol{\Lambda}_k]_{l,l}=\operatorname{max}\left\{[\boldsymbol{\Sigma}_k]_{l,l}-\frac{\gamma}{2},0\right\}, \,\forall l 
\end{equation}
and $\gamma$ is the smallest nonnegative value satisfying $g\left(\gamma\right)\leq P_T$, where $\gamma \mapsto g(\gamma) = \sum_{k=1}^K \sum_{l=1}^K \max \left\{ [ {\mathbf \Sigma}_k ]_{l,l} - \frac{\gamma}{2}, 0 \right\}$ is a piecewise linear, decreasing function on $\gamma \ge 0$.
\end{proposition}
\begin{proof}
    See Appendix \ref{app_pro_subproblem_X} in the Supplementary Material.
\end{proof}

\begin{proposition}\label{pro_subproblem_Y}
    Let $\widetilde{\mathbf{Y}}=\mathbf{U}\boldsymbol{\Sigma}\mathbf{U}^H$ denote the eigenvalue decomposition of $\widetilde{\mathbf{Y}}$. Then, the optimal solution to the $\mathbf{Y}$-subproblem \eqref{eq_Y_subproblem} is given by
    \begin{equation}\label{eq_update_Y}
        \mathbf{Y}=\mathbf{U}\boldsymbol{\Lambda}\mathbf{U}^H,
    \end{equation}
    where $\boldsymbol{\Lambda}$ is a diagonal matrix with $[\boldsymbol{\Lambda}]_{k,k}$ being the unique positive root of the equation $x^3-[\boldsymbol{\Sigma}]_{k,k} x^2-\tau=0$.
\end{proposition}
\begin{proof}
    See Appendix \ref{app_pro_subproblem_Y} in the Supplementary Material.
\end{proof}

\begin{proposition}\label{pro_subproblem_Z}
    Let $\widetilde{\mathbf{Z}}=\mathbf{U}\boldsymbol{\Sigma}\mathbf{U}^H$ denote the eigenvalue decomposition of $\widetilde{\mathbf{Z}}$. Then, the optimal solution to the $\mathbf{Z}$-subproblem \eqref{eq_Z_subproblem} is given by
    \begin{equation}\label{eq_update_Z1}
        \mathbf{Z}=\mathbf{U}\boldsymbol{\Lambda}\mathbf{U}^H,
    \end{equation}
    where $\boldsymbol{\Lambda}$ is a diagonal matrix with
    \begin{equation}\label{eq_update_Z2}
   [\boldsymbol{\Lambda}]_{k,k}=\operatorname{max}\left\{[\boldsymbol{\Sigma}]_{k,k}-\lambda,0\right\}
\end{equation}
and $\lambda$ is the unique root of the equation
\begin{equation}\label{eq_update_Z3}
    \lambda\left(\!P_T-\sum_{k=1}^K\operatorname{max}\left\{[\boldsymbol{\Sigma}]_{k,k}\!-\lambda,0\right\}\right)^2\!-\tau(N_t-K)^2=0.
\end{equation}
\end{proposition}
\begin{proof}
    See Appendix \ref{app_pro_subproblem_Z} in the Supplementary Material.
\end{proof}


\subsubsection{$\boldsymbol{\lambda}$-subproblem}
Upon examining the third line of \eqref{eq_BALM}, we observe that the key step there involves explicitly expressing the inverse of the matrix $\mathbf{D}\mathbf{D}^H+\delta\mathbf{I}_{K+2K^2}$. At first glance, computing this inverse would incur a complexity of $\mathcal{O}((K + K^2)^3)$. However, by exploiting the block structure of the matrices $\mathbf{A}$, $\mathbf{B}$, and $\mathbf{C}$, we can reduce the complexity to $\mathcal{O}(K^3)$. We provide the details below.

First, through straightforward calculations, we obtain
\begin{equation}\label{eq_DD}
\begin{aligned}
&\mathbf{D}\mathbf{D}^H+\delta\mathbf{I}_{K+2K^2}=\\&\left(\begin{array}{ccc}
\mathbf{A}_1\mathbf{A}_1^{H}+ |\mathbf{H}^H\mathbf{H}|^2+\delta \mathbf{I}_K & \mathbf{T}_{12} & \mathbf{B}_1 \\
\mathbf{T}_{12}^H & \alpha\mathbf{I}_{K^2} & -\mathbf{I}_{K^2} \\
\mathbf{B}_1^H & -\mathbf{I}_{K^2}& \beta\mathbf{I}_{K^2}
\end{array}\right),
\end{aligned}
\end{equation}
where $\mathbf{A}_1 \mathbf{A}_1^{H}=\operatorname{Diag}\left(\rho_1^2\left\|\mathbf{h}_1\right\|^4, \ldots, \rho_K^2\left\|\mathbf{h}_K\right\|^4\right)$ is a diagonal matrix, $|\mathbf{H}^H\mathbf{H}|^2$ denotes the matrix obtained by taking square modulus of each element in $\mathbf{H}^H\mathbf{H}$, $\alpha= K+1+\delta$, $\beta= \delta+2$, and
\begin{equation}\notag
\mathbf{T}_{12}=\left(\begin{array}{c}
\left(\rho_1+1\right) \operatorname{vec}\left(\tilde{\mathbf{Q}}_1\right)^{H} \\
\vdots \\
\left(\rho_K+1\right) \operatorname{vec}\left(\tilde{\mathbf{Q}}_K\right)^{H}
\end{array}\right).
\end{equation}

\noindent Now, let $\mathbf{S}_{12}= [\mathbf{T}_{12},\mathbf{B}_{1}]$ and $\mathbf{S}_{22} $ denote the $2\times 2$ block matrix in the lower-right corner of \eqref{eq_DD}. Although the dimension of $\mathbf{S}_{22}$ is $2K^2 \times 2K^2$, which appears large, its inverse can be computed efficiently due to its special structure. By applying the block matrix inversion formula in \cite[Section 9.1.3]{petersen2008matrix}, we obatin
\begin{equation}\begin{aligned}\label{eq_DDH_inverse} 
   &\left(\mathbf{D}\mathbf{D}^H+\delta\mathbf{I}_{K+2K^2}\right)^{-1}\\ &\!\!=\!\!\left(\!\!\!\begin{array}{ccc}
\mathbf{L}^{-1}& -\mathbf{L}^{-1}\mathbf{V}_{1} & \!\!\! -\mathbf{L}^{-1}\mathbf{V}_{2}\\
-\mathbf{V}_{1}^H\mathbf{L}^{-1} & \!\!\kappa\beta\mathbf{I}_{K^2}{+}\mathbf{V}_{1}^H\mathbf{L}^{-1}\mathbf{V}_{1}&\!\!\kappa\mathbf{I}_{K^2}{+}\mathbf{V}_{1}^H\mathbf{L}^{-1}\mathbf{V}_{2} \\-\mathbf{V}_{2}^H\mathbf{L}^{-1}&\!\kappa\mathbf{I}_{K^2}{+}\mathbf{V}_{2}^H\mathbf{L}^{-1}\mathbf{V}_{1}&\!\!\kappa\alpha\mathbf{I}_{K^2}{+}\mathbf{V}_{2}^H\mathbf{L}^{-1}\mathbf{V}_{2}
\end{array}\!\!\!\right),
\end{aligned}\end{equation}
where
\begin{equation}\notag
\begin{aligned}
   \mathbf{L}= &\mathbf{A}_1\mathbf{A}_1^{H}+ |\mathbf{H}^H\mathbf{H}|^2+\delta \mathbf{I}_K- \mathbf{S}_{12} \mathbf{S}_{22}^{-1} \mathbf{S}_{12}^H\\
   =& \delta \mathbf{I}_K+|\mathbf{H}^H\mathbf{H}|^2\odot\left(\operatorname{Diag}(\boldsymbol{\rho}\odot \boldsymbol{\rho})+\mathbf{1}_K \mathbf{1}_K^H-\kappa\mathbf{P}\right),
   \end{aligned}
\end{equation}
$\boldsymbol{\rho}=[\rho_1,\rho_2,\ldots,\rho_K]^T$, 
\begin{equation}\notag
\begin{aligned}
\mathbf{P}=&\left(\delta+2\right)\left(\boldsymbol{\rho}+\mathbf{1}_K\right)\left(\boldsymbol{\rho}+\mathbf{1}_K\right)^H-\left(\boldsymbol{\rho}+\mathbf{1}_K\right)\mathbf{1}_K^H\\ &-\mathbf{1}_K\left(\boldsymbol{\rho}+\mathbf{1}_K\right)^H+\left(\delta+K+1\right)\mathbf{1}_K\mathbf{1}_K^H,
    \end{aligned}
\end{equation}
and
\begin{equation}\notag
    \kappa=\frac{1}{\left(\delta+K+1\right)\left(\delta+2\right)-1}.
\end{equation}
In addition, we have
\begin{equation}\notag
    \mathbf{V}_{1}=\kappa\beta\mathbf{T}_{12}+\kappa\mathbf{B}_1=\boldsymbol{\Theta}_1\mathbf{B}_1,
\end{equation}
\begin{equation}\notag
    \mathbf{V}_{2}=\kappa\mathbf{T}_{12}+\kappa\alpha\mathbf{B}_1=\boldsymbol{\Theta}_2\mathbf{B}_1,
\end{equation}
where $\boldsymbol{\Theta}_1$ and $\boldsymbol{\Theta}_2$ are diagonal matrices whose $k$-th diagonal elements are given by $[\boldsymbol{\Theta}_1]_{k,k}=\kappa(1-\beta\rho_k-\beta)$ and $[\boldsymbol{\Theta}_2]_{k,k}=\kappa(\alpha-\rho_k-1)$, respectively.
 Note that only a low-dimensional inversion $\mathbf{L}^{-1}$ is involved in \eqref{eq_DDH_inverse}, and its computational complexity is $\mathcal{O}(K^3)$. In addition, this inverse operation only needs to be computed once and can be reused throughout all iterations.

With the above preparations, after obtaining the iterate ${\mathbf u}^{t+1}$, we can obtain the iterates ${\mathbf p}^{t+1}$ and ${\boldsymbol \lambda}^{t+1}$ as follows. First, define 
$r_k^{t+1}=\rho_k\operatorname{tr}\left(\tilde{\mathbf{Q}}_{k}\left(2\mathbf{X}_k^{t+1}-\mathbf{X}_k^{t}\right)\right)-\operatorname{tr}\left(\tilde{\mathbf{Q}}_{k}\left(2\mathbf{Y}^{t+1}-\mathbf{Y}^{t}\right)\right) -\sigma^2_C, \, \forall k$. Then, define $\mathbf{R}_1^{t+1}=\sum_{k=1}^K\left(2\mathbf{X}_k^{t+1}-\mathbf{X}_k^{t}\right)-\left(2\mathbf{Y}^{t+1}-\mathbf{Y}^{t}\right)$ and $\mathbf{R}_2^{t+1}=2\mathbf{Y}^{t+1}-\mathbf{Y}^{t}-2\mathbf{Z}^{t+1}+\mathbf{Z}^{t}$. We then have ${\mathbf p}^{t+1}=[({\mathbf r}^{t+1})^T,\operatorname{vec}(\mathbf{R}_1^{t+1})^T,\operatorname{vec}(\mathbf{R}_2^{t+1})^T]^T$.

Next, define $\tilde{\mathbf{H}}=\tilde{\mathbf{U}}^H\mathbf{H}$ and
\begin{equation}\label{eq_update_lambda}
\begin{aligned}
\boldsymbol{\mu}^{t+1}=\boldsymbol{\mu}^{t} + (\tau \mathbf{L})^{-1}\left(\mathbf{r}^{t+1}+\boldsymbol{\Theta}_1\operatorname{diag}\left(\tilde{\mathbf{H}}^H\mathbf{R}_1^{t+1}\tilde{\mathbf{H}}\right)\right.\\
    ~~+\left.\boldsymbol{\Theta}_2\operatorname{diag}\left(\tilde{\mathbf{H}}^H\mathbf{R}_2^{t+1}\tilde{\mathbf{H}}\right)\right),
    \end{aligned}
\end{equation}
\begin{equation}\label{eq_update_lambda_1}
\begin{aligned}
\boldsymbol{\Omega}_1^{t+1}=\,\,&\boldsymbol{\Omega}_1^{t}+\tau^{-1}\left(\kappa\beta\mathbf{R}_1^{t+1}+\kappa\mathbf{R}_2^{t+1}\right)\\
&~~~~~~~~~+\tilde{\mathbf{H}}\boldsymbol{\Theta}_1\operatorname{Diag}\left(\boldsymbol{\mu}^{t+1}-\boldsymbol{\mu}^t\right)\tilde{\mathbf{H}}^H,
\end{aligned}
\end{equation}
\begin{equation}\label{eq_update_lambda_2}
\begin{aligned}
\boldsymbol{\Omega}_2^{t+1}=\,\,&\boldsymbol{\Omega}_2^{t}+\tau^{-1}\left(\kappa\mathbf{R}_1^{t+1}+\kappa\alpha\mathbf{R}_2^{t+1}\right)\\
&~~~~~~~~~+\tilde{\mathbf{H}}\boldsymbol{\Theta}_2\operatorname{Diag}\left(\boldsymbol{\mu}^{t+1}-\boldsymbol{\mu}^t\right)\tilde{\mathbf{H}}^H.
\end{aligned}
\end{equation}
We then have ${\boldsymbol \lambda^{t+1}{=}[\operatorname{vec}(\boldsymbol{\mu}^{t+1}),\operatorname{vec}(\boldsymbol{\Omega}_1^{t+1})^T,\operatorname{vec}(\boldsymbol{\Omega}_2^{t+1})^T]^T}$.



The overall procedure is summarized in Algorithm \ref{alg_RBAL}. Once the output of the R-BAL method is obtained, Theorem \ref{thm_extract_rank1} can be applied to extract a rank-one solution, which in turn yields an optimal solution to the original problem \eqref{eq_extended_mu_isac_problem}.  The dominant computational cost of  Algorithm \ref{alg_RBAL} arises from the updates of the primal variables $\mathbf{X}^{t+1}$, $\mathbf{Y}^{t+1}$,  $\mathbf{Z}^{t+1}$ and the dual variables $\boldsymbol{\mu}^{t+1}$, $\boldsymbol{\Omega}_1^{t+1}$,  $\boldsymbol{\Omega}_2^{t+1}$. Specifically, updating the primal variables requires $K+2$ eigenvalue decompositions of $K\times K$ Hermitian matrices, resulting in a total complexity of $\mathcal{O}(K^4)$ \cite{golub1996matrix}. Meanwhile, updating the dual variables involves inverting a $K\times K$ matrix $\mathbf{L}$, which incurs a complexity of $\mathcal{O}(K^3)$. Overall, the per-iteration computational complexity of the proposed R-BAL method is $\mathcal{O}(K^4)$.
\begin{algorithm}[!tb]
    \renewcommand{\algorithmicrequire}{\textbf{Input:}}
    \renewcommand{\algorithmicensure}{\textbf{Output:}}
    \caption{R-BAL Method for Solving Problem \eqref{eq_extended_mu_isac_problem_wc_lds_sdr_brief11}}  \label{alg_RBAL}
    \begin{algorithmic}[1]
    \REQUIRE Initialize $\mathbf{X}^0$, $\boldsymbol{\mu}^0$, $\boldsymbol{\Omega}_1^0$, and $\boldsymbol{\Omega}_2^0$. Set the stepsize $\tau$.
    \FOR{$t=0,1,\ldots$}
    \STATE $\widetilde{\mathbf{X}}_{k}^{t} = \mathbf{X}_k^t-\tau\left(\rho_k\mu^t_k\tilde{\mathbf{Q}}_k+\boldsymbol{\Omega}_1^t\right),\forall k$;\\
    \STATE  Update $\mathbf{X}^{t+1} $ by \eqref{eq_update_X1} and \eqref{eq_update_X2};
    \STATE   $\widetilde{\mathbf{Y}}^{t} = \mathbf{Y}^t+\tau\left(\sum_{k=1}^K\mu_k^t\tilde{\mathbf{Q}}_k+\boldsymbol{\Omega}_1^t-\boldsymbol{\Omega}_2^t\right)$;
     \STATE Update $\mathbf{Y}^{t+1}$ by \eqref{eq_update_Y};
     \STATE $\widetilde{\mathbf{Z}}^{t} = \mathbf{Z}^t-\tau\boldsymbol{\Omega}_2^t$;
     \STATE Update $\mathbf{Z}^{t+1}$ by \eqref{eq_update_Z1}--\eqref{eq_update_Z3};
     \STATE  Update $\mathbf{r}^{t+1} $ by $r_k^{t+1}=\rho_k\operatorname{tr}\left(\tilde{\mathbf{Q}}_{k}\left(2\mathbf{X}_k^{t+1}-\mathbf{X}_k^{t}\right)\right)-\operatorname{tr}\left(\tilde{\mathbf{Q}}_{k}\left(2\mathbf{Y}^{t+1}-\mathbf{Y}^{t}\right)\right) -\sigma^2_C, \, \forall k$;
     \STATE Update ${\mathbf R}_1^{t+1}$ by $\mathbf{R}_1^{t+1}=\sum_{k=1}^K\left(2\mathbf{X}_k^{t+1}-\mathbf{X}_k^{t}\right)-\left(2\mathbf{Y}^{t+1}-\mathbf{Y}^{t}\right)$;
     \STATE Update ${\mathbf R}_2^{t+1}$ by $\mathbf{R}_2^{t+1}=2\mathbf{Y}^{t+1}-\mathbf{Y}^{t}-2\mathbf{Z}^{t+1}+\mathbf{Z}^{t}$;
     \STATE Update $\boldsymbol{\mu}^{t+1}$, $\boldsymbol{\Omega}_1^{t+1}$, and $\boldsymbol{\Omega}_2^{t+1}$ by \eqref{eq_update_lambda}, \eqref{eq_update_lambda_1}, and \eqref{eq_update_lambda_2}, respectively.
    \ENDFOR
    \ENSURE $\{ {\mathbf X}_k \}_{k=1}^K$.
    \end{algorithmic}
\end{algorithm}




\section{Numerical Results}
\label{sec:numerical_results}

In this section, we provide numerical results to demonstrate the efficiency of our proposed low-complexity algorithms, which exploit the low-dimensional structures established in Theorem \ref{theorem_low_dimen}. We compare the performance of the following four algorithms:
\begin{itemize}
    \item \textbf{IPM}: The interior-point method applied to the original high-dimensional problem \eqref{eq_extended_mu_isac_problem_sdr}, as in \cite{liu2021cramer}.
    \item \textbf{R-IPM}: The interior-point method applied to our proposed reduced problem \eqref{eq_extended_mu_isac_problem_wc_lds_sdr_brief}.
    \item \textbf{BAL}: The balanced augmented Lagrangian method applied to the original high-dimensional problem \eqref{eq_extended_mu_isac_problem_sdr}.
    \item \textbf{R-BAL}: Our proposed balanced augmented Lagrangian method (Algorithm \ref{alg_RBAL}) applied to the problem with low-dimensional structure \eqref{eq_extended_mu_isac_problem_wc_lds_sdr_brief11}.
\end{itemize}

\subsection{Experiment Setup}

We consider a downlink massive MIMO ISAC system, where the channel matrix $\mathbf{H}$ follows an i.i.d. complex Gaussian distribution with zero mean and unit variance. All simulation results are averaged over 100 independent Monte Carlo trials to mitigate the effects of random channel realizations. As a default setting, the transmit power budget is set to $P_T = 20$\,dBm. The SINR threshold is $\Gamma_k = 10$\,dB for all users, and the noise power is $\sigma_C^2 = 0$\,dBm. For the BAL and R-BAL methods, the stopping criterion is set based on a tolerance of $10^{-9}$ for the constraint-violation norms. The simulations were conducted in MATLAB R2024b on a machine equipped with an Intel Core i9-10900X CPU and 64 GB of RAM.

We investigate the following two main scenarios:
\begin{enumerate}
    \item \emph{$K$-sweep}: We fix the number of BS antennas $N_t = 64$ and vary the number of users $K$ from 4 to 16.
    \item \emph{$N_t$-sweep}: We fix the number of users $K = 8$ and vary the number of BS antennas $N_t$ from 16 to 128. 
\end{enumerate}

\subsection{CRB Performance}
\begin{figure}[!t]
    \centering
    \begin{subfigure}[t]{\linewidth}
        \centering
        \includegraphics[width=0.66\linewidth]{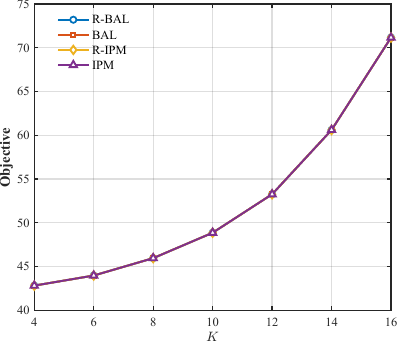}
        \caption{CRB versus the number of users $K$, with $N_t=64$.}
        \label{fig:objective_vs_K}
    \end{subfigure}
    
    \begin{subfigure}[t]{\linewidth}
        \centering
        \includegraphics[width=0.66\linewidth]{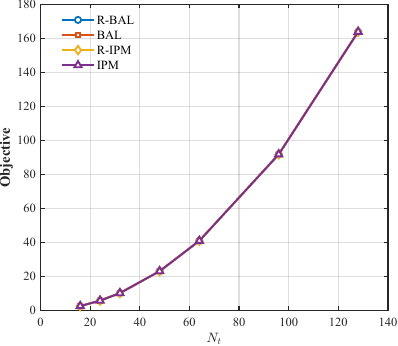}
        \caption{CRB versus the number of transmitting antennas $N_t$, with $K=8$.}
        \label{fig:objective_vs_Nt}
    \end{subfigure}
    \caption{CRB performance under two scenarios. (a) $N_t=64$ and varying $K$; (b) $K=8$ and varying $N_t$. }
    \label{fig:objective_combined}
\end{figure}
First, we evaluate the optimality of the proposed algorithms by comparing their achieved objective values. The results are shown in Fig. \ref{fig:objective_combined}. Specifically, Fig.~\ref{fig:objective_vs_K} and Fig.~\ref{fig:objective_vs_Nt} show the CRB performance  under varying number of users  $K$ and varying number of BS antennas $N_t$, respectively. As shown in both figures, the objective values achieved by all four algorithms are identical across all tested configurations. This result empirically validates a key theoretical finding of our work: The low-dimensional structures found in Theorem \ref{theorem_low_dimen} give rise to a problem reformulation that does not incur any optimality loss. All algorithms consistently converge to the same globally optimal solution.

\begin{table*}[!t]
\centering
\caption{Average runtime (s) versus the number of users $K$, with $N_t=64$.}
\label{tab:runtime_vs_K}
\begin{tabular}{@{}lccccccc@{}}
\toprule
Algorithm & $K=4$ & $K=6$ & $K=8$ & $K=10$ & $K=12$ & $K=14$ & $K=16$ \\
\midrule
IPM  \cite{liu2021cramer}    & 411.49 & 502.05 & 657.59 & 669.93 & 688.94 & 724.21 & 787.53 \\
R-IPM                        & 1.18  & 1.72  & 2.95  & 5.13  & 6.35  & 7.26  & 8.67  \\
BAL  & 6.53  & 8.96  & 10.61 & 11.65 & 12.145 & 12.97  & 13.07  \\
R-BAL                        & 0.15  & 0.19  & 0.24  & 0.30  & 0.39  & 0.54  & 0.81  \\
\bottomrule
\end{tabular}
\end{table*}

\begin{table*}[!t]
\centering
\caption{Average runtime (s) versus the number of antennas $N_t$, with $K=8$.}
\label{tab:runtime_vs_Nt}
\begin{tabular}{@{}lccccccc@{}}
\toprule
Algorithm & $N_t=16$ & $N_t=24$ & $N_t=32$ & $N_t=48$ & $N_t=64$ & $N_t=96$ & $N_t=128$ \\
\midrule
IPM \cite{liu2021cramer}  & 4.62  & 17.09 & 46.97 & 281.18 & 678.05 & 1967.30 & 2859.75 \\
R-IPM  & 2.968  & 2.88  & 2.85  & 2.88  & 2.95  & 2.89   & 3.03   \\
BAL & 0.67  & 1.27  & 2.54  & 5.09  & 9.55  & 20.10  & 36.97  \\
R-BAL  & 0.25  & 0.26  & 0.25  & 0.25  & 0.25  & 0.26   & 0.27  \\
\bottomrule
\end{tabular}
\end{table*}

\begin{figure*}[!t]
    \centering
    \begin{subfigure}[t]{0.32\textwidth}
        \centering
        \includegraphics[width=\linewidth]{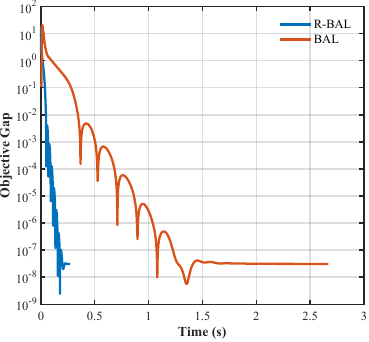}
        \caption{Gap vs time, $N_t=32$.}
    \end{subfigure}\hfill
    \begin{subfigure}[t]{0.32\textwidth}
        \centering
        \includegraphics[width=\linewidth]{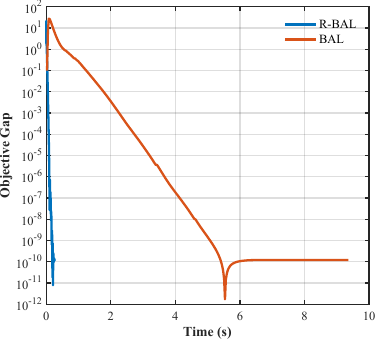}
        \caption{Gap vs time, $N_t=64$.}
    \end{subfigure}\hfill
    \begin{subfigure}[t]{0.32\textwidth}
        \centering
        \includegraphics[width=\linewidth]{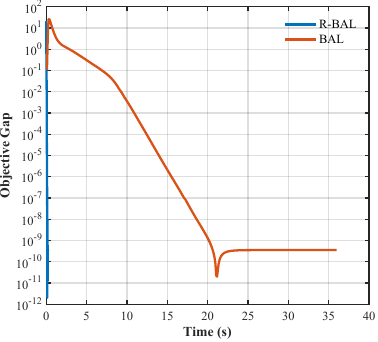}
        \caption{Gap vs time, $N_t=128$.}
    \end{subfigure}
    \caption{Optimality gap versus wall-clock time under different $N_t$, with $K=8$ and $P_T=15$ dBm. }
    \label{fig:gap_vs_time_nt_sweep}
\end{figure*}

\subsection{Computational Efficiency}

Next, we evaluate the computational efficiency of the algorithms by comparing their average runtime. The results are summarized in Table~\ref{tab:runtime_vs_K} and Table~\ref{tab:runtime_vs_Nt}.

As shown in Table~\ref{tab:runtime_vs_K}, which presents the runtime comparison in the $K$-sweep scenario, the proposed R-BAL method demonstrates the highest efficiency, consistently outperforming all other methods by a significant margin. Although R-IPM is slower than R-BAL, it still achieves a considerable speedup over the original IPM. This underscores the advantage of leveraging the low-dimensional structure even within a conventional IPM solver.

The benefits of our proposed methods are even more pronounced in the $N_t$-sweep scenario, as illustrated in Table~\ref{tab:runtime_vs_Nt}. Here, the runtimes of both IPM and BAL increase sharply with the number of antennas $N_t$, rendering them impractical for massive MIMO systems. By contrast, the runtimes of R-IPM and R-BAL remain consistently low and nearly invariant with increasing $N_t$. This behavior is in full agreement with our theoretical complexity analysis: While the complexities of IPM and BAL scale strongly with $N_t$, those of R-IPM ($\mathcal{O}(K^{10})$) and R-BAL ($\mathcal{O}(K^4)$) are independent of $N_t$. This confirms the exceptional scalability of our low-dimensional framework in large-scale antenna systems. In particular, the R-BAL method delivers the best performance, achieving a speedup of 10000× for large antenna arrays (e.g., $N_t>100$) when compared to the original IPM \cite{liu2021cramer}.

\subsection{Convergence Behavior}

Fig.~\ref{fig:gap_vs_time_nt_sweep} and Fig.~\ref{fig:res_vs_time_nt_sweep} illustrate the convergence of the algorithms in terms of the optimality gap and constraint violation versus wall-clock time, respectively.  Here, the optimality gap is defined as the absolute difference between the objective value of a point and the optimal objective value obtained via the IPM solver. The constraint violation is measured by the Euclidean norm of the constraint residuals of a point. The results correspond to scenarios with a fixed number of users ($K=8$) and an increasing number of antennas ($N_t \in \{32, 64, 128\}$).

The results clearly demonstrate the superior performance and scalability of the proposed R-BAL method. As shown in Fig. \ref{fig:gap_vs_time_nt_sweep} and Fig.~\ref{fig:res_vs_time_nt_sweep}, the convergence curves for both optimality gap and constraint violation remain nearly identical across different values of $N_t$. This highlights that the time-to-accuracy of the proposed R-BAL method is independent of the number of BS antennas, which empirically confirms the benefits of the low-dimensional formulation and its $\mathcal{O}(K^4)$ complexity. By contrast, the high-dimensional BAL baseline shows a noticeable slowdown as $N_t$ increases, with its curves shifting to the right. This behavior is consistent with the runtime data in Table~\ref{tab:runtime_vs_Nt} and underscores the computational challenges of operating in the original high-dimensional space. Furthermore, although the BAL method is ultimately convergent, it exhibits significant oscillations. This behavior arises from the use of a fixed step size, which prevents the algorithm from fully satisfying the constraints prior to convergence. By contrast, the R-BAL method, despite also using a fixed step size, exhibits a markedly smoother convergence process. This can be attributed to its ability to leverage the inherent low-dimensional structures of the problem.

\begin{figure*}[!t]
    \centering
    \begin{subfigure}[t]{0.32\textwidth}
        \centering
        \includegraphics[width=\linewidth]{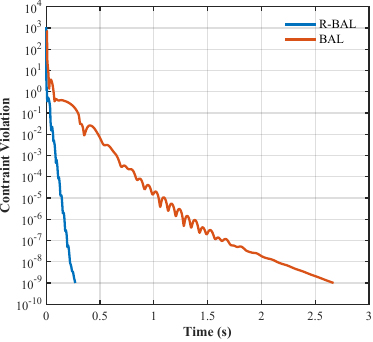}
        \caption{Constraint violation vs time, $N_t=32$.}
    \end{subfigure}\hfill
    \begin{subfigure}[t]{0.32\textwidth}
        \centering
        \includegraphics[width=\linewidth]{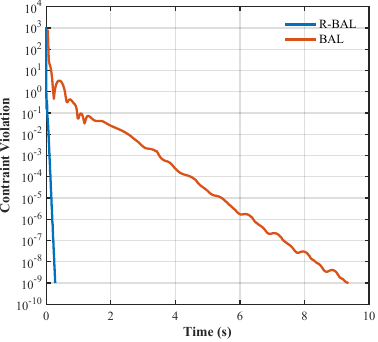}
        \caption{Constraint violation vs time, $N_t=64$.}
    \end{subfigure}\hfill
    \begin{subfigure}[t]{0.32\textwidth}
        \centering
        \includegraphics[width=\linewidth]{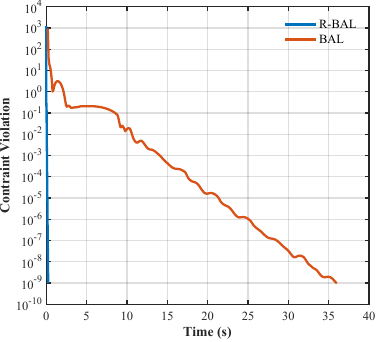}
        \caption{Constraint violation vs time, $N_t=128$.}
    \end{subfigure}
    \caption{Constraint violation versus wall-clock time under different $N_t$, with $K=8$ and $P_T=15$ dBm.}
    \label{fig:res_vs_time_nt_sweep}
\end{figure*}



In summary, the numerical results show that our proposed R-BAL method is able to exploit the inherent low-dimensional structure of the ISAC beamforming problem and achieves state-of-the-art computational efficiency and scalability while preserving optimality.

\section{Conclusion}
In this paper, we showed that a transmit beamforming design problem that arises in massive MIMO ISAC systems possesses a low-dimensional structure. This structure enables an equivalent reformulation of the problem, in which the dimension depends on the number of users rather than the number of BS antennas. By exploiting this reformulation, we can reduce the computational complexity of standard IPMs from $\mathcal{O}(N_t^{6.5}K^{3.5})$ to $\mathcal{O}(K^{10})$. Moreover, we introduced a R-BAL method that further reduces the complexity to $\mathcal{O}(K^4)$. Lastly, we presented numerical results to demonstrate that the proposed approach can yield a speedup of up to $10000\times$ over existing ones in CPU runtime.









\appendices

\section{Proof of Proposition \ref{pro_infea_condition} }\label{app_pro_infea_condition}

On one hand, a feasible solution to problem \eqref{eq_extended_mu_isac_problem} corresponds to a feasible rank-one solution to problem \eqref{eq_extended_mu_isac_problem_sdr}. On the other hand, given a feasible solution to problem \eqref{eq_extended_mu_isac_problem_sdr}, a feasible rank-one solution to problem \eqref{eq_extended_mu_isac_problem} can be constructed; see \cite{liu2020joint}. Therefore, to understand when problem \eqref{eq_extended_mu_isac_problem_sdr} is feasible, it suffices to analyze the feasibility of problem \eqref{eq_extended_mu_isac_problem}.

Towards that end, consider the following power minimization problem:
    \begin{equation}\label{eq_power_mini}
\begin{aligned}
 \min_{\{\mathbf{w}_k\}_{k=1}^K,\mathbf{W}_{A}} ~&\sum_{k=1}^{K} \|\mathbf{w}_k\|^2+ \|\mathbf{W}_A\|_F^2\\
 \text{ s.t. } ~~~~~&\frac{|\mathbf{h}_{k}^{H}\mathbf{w}_{k}|^{2}}{\sum_{i=1,i\neq k}^{K}|\mathbf{h}_{k}^{H}\mathbf{w}_{i}|^{2} + \| \mathbf{h}_{k}^{H}\mathbf{W}_{A} \|^{2} + \sigma_{C}^{2}} \geq\Gamma_k, ~ \forall k.
\end{aligned}
\end{equation}
Let $P^{\ast}$ denote the optimal objective value of problem \eqref{eq_power_mini}. Obviously, problem \eqref{eq_extended_mu_isac_problem} is infeasible if and only if $P_T<P^{\ast}$. Therefore, it suffices to prove that $P^{\ast}=P_{low}$, where $P_{low}$ is defined in \eqref{eq_p_low}.

Note that every optimal solution $( \{\mathbf{w}_k^*\}_{k=1}^K, {\mathbf W}_A^* )$ to problem \eqref{eq_power_mini} must satisfy $\mathbf{W}_{A}^{\ast}=\mathbf{0}$. Hence, the optimal value of problem \eqref{eq_power_mini} is the same as that of
 \begin{equation}\label{eq_power_mini_degenerate}
\begin{aligned}
 \min_{\{\mathbf{w}_{k}\}_{k=1}^K} ~~&\sum_{k=1}^{K} \|\mathbf{w}_k\|^2\\
 \text{ s.t. } ~~~~&\frac{|\mathbf{h}_{k}^{H}\mathbf{w}_{k}|^{2}}{\sum_{i=1,i\neq k}^{K}|\mathbf{h}_{k}^{H}\mathbf{w}_{i}|^{2} + \sigma_{C}^{2}} \geq\Gamma_k, ~ \forall k,
\end{aligned}
\end{equation}
which is exactly the classical power minimization problem.

By utilizing the low-dimensional subspace structure established in \cite[Theorem 5]{zhao2025universal}, problem \eqref{eq_power_mini_degenerate} can be equivalently reformulated as the following lower-dimensional problem:
\begin{equation}\label{eq_power_mini_reduced}
\begin{aligned}
 \min_{\{\mathbf{x}_k\}_{k=1}^K} ~~&\sum_{k=1}^{K} \|\mathbf{H}\mathbf{x}_k\|_2^2 \\
 \text{ s.t. } ~~~~&\frac{\left|\widebar{\mathbf{h}}_k^H \mathbf{x}_k\right|^2}{\sum_{i \neq k}^{K}\left|\widebar{\mathbf{h}}_k^H \mathbf{x}_i\right|^2+\sigma_C^2} \geq\Gamma_k,~ \forall k.
\end{aligned}
\end{equation} 
Notably, for each $k \in [K]$, the vector $\mathbf{x}_k$ has dimension $K$, whereas the vector $\mathbf{w}_k$ has dimension $N_t$.

Following \cite{wiesel2005linear}, it can be shown that the KKT conditions associated with problem \eqref{eq_power_mini_reduced} are both necessary and sufficient for optimality. Consider the Lagrangian function $\mathcal{L}$ of problem \eqref{eq_power_mini_reduced} given by
\begin{equation}\notag
\begin{aligned}
 \mathcal{L}\left(\mathbf{X}, \boldsymbol{\lambda}\right) =&\sum_{k=1}^K\left\|\mathbf{H}\mathbf{x}_k\right\|_2^2+\sum_{k=1}^K \lambda_k\Bigg(\sum_{i \neq k} \frac{1}{\sigma_C^2}\left|\widebar{\mathbf{h}}_k^H \mathbf{x}_i\right|^2\\
&+1-\frac{1}{\gamma_k \sigma_C^2}\left|\widebar{\mathbf{h}}_k^H \mathbf{x}_k\right|^2\Bigg),
\end{aligned}
\end{equation}
where $\boldsymbol{\lambda}=[\lambda_1,\lambda_2,\ldots,\lambda_K]$ with $\lambda_k\geq 0$ being the Lagrangian multiplier associated with the $k$-th SINR constraint.
The dual function is given by $\min_{\mathbf{X}}\mathcal{L}\left(\mathbf{X}, \boldsymbol{\lambda}\right)$, 
which can be written in closed form as $\sum_{i=1}^K\lambda_i$ by checking the KKT conditions.
Since the strong duality holds, we have $\sum_{i=1}^K \lambda_i^{\ast} = P^{\ast}$, where $\{ \lambda_k^* \}_{k=1}^K$ are the optimal multipliers.
 Following a derivation analogous to that in \cite{bjornson2014optimal}, these multipliers can be computed via the fixed-point iteration given in \eqref{eq_lambda_fixedpoint}, thereby completing the proof.

\section{Proof of Theorem \ref{theorem_low_dimen} }\label{app_theorem_low_dimen}

    In order to derive the optimal beamforming structure, we begin by considering the Lagrangian function $\mathcal{L}$ of problem \eqref{eq_extended_mu_isac_problem_sdr} given by
\begin{equation}\notag
    \begin{aligned}
&\mathcal{L}( {\mathbf W}, {\boldsymbol \mu}, \omega, {\mathbf \Theta} )\\
=&\operatorname{tr}\left(\mathbf{R}_{W}^{-1}\right)+\omega\left( \operatorname{tr}\left(\mathbf{R}_{W}\right)- P_T\right)-\sum_{j=1}^{K+1}\operatorname{tr}\left(\boldsymbol{\Theta}_j\mathbf{W}_j\right)\\
&-\sum_{k=1}^{K}\mu_k\Big(\rho_k\operatorname{tr}\left(\mathbf{Q}_{k}\mathbf{W}_{k}\right) -\sum_{j=1}^{K+1}\operatorname{tr}\left(\mathbf{Q}_{k}\mathbf{W}_{j}\right)-\sigma^2_C\Big),
    \end{aligned}
\end{equation}
where $\mathbf{R}_{W}= \sum_{k=1}^{K+1}\mathbf{W}_k$ and $\mu_k \ge 0$ for $k \in [K]$, $\omega \ge 0$, ${\mathbf \Theta}_k \succeq {\mathbf 0}$ for $k \in [K+1]$ are the Lagrangian multipliers. The KKT conditions associated with problem \eqref{eq_extended_mu_isac_problem_sdr} are given by
\begin{subequations}\label{eq_kkt_conditions}
\begin{align}
& -\mathbf{R}_{W}^{-2}+\omega\mathbf{I}_{N_t}+\sum_{j=1}^K\mu_j\mathbf{Q}_{j}-\mu_k\rho_k\mathbf{Q}_{k}=\boldsymbol{\Theta}_k,\forall k,\label{eq_kkt_w1_gra_mu}\\
& -\mathbf{R}_{W}^{-2}+\omega\mathbf{I}_{N_t}+\sum_{j=1}^{K}\mu_j\mathbf{Q}_{j}=\boldsymbol{\Theta}_{K+1},\label{eq_kkt_wA_gra_mu}\\
&\mu_k\Big(\rho_k\operatorname{tr}\left(\mathbf{Q}_{k}\mathbf{W}_{k}\right){-}\sum_{j=1}^{K+1}\operatorname{tr}\left(\mathbf{Q}_{k}\mathbf{W}_{j}\right)-\sigma^2_C\Big)=0,\forall k,\label{eq_kkt_comp}\\
&\omega\left( \operatorname{tr}\left(\mathbf{R}_{W}\right)- P_T\right)=0,\label{eq_kkt_comp_power}\\
&\boldsymbol{\Theta}_k\mathbf{W}_k=\mathbf{0},\forall k\in [K+1],\label{eq_kkt_slackness}\\
&\rho_k\operatorname{tr}\left(\mathbf{Q}_{k}\mathbf{W}_{k}\right)-\sum_{j=1}^{K+1}\operatorname{tr}\left(\mathbf{Q}_{k}\mathbf{W}_{j}\right)\geq\sigma^2_C,\forall k, \label{eq_kkt_w1_fea_W}\\
& \operatorname{tr}\left(\mathbf{R}_{W}\right) \leq P_T,\,\,\,\mathbf{W}_{k}\succeq\mathbf{0},\forall k\in [K+1],\\
&\mu_k\geq 0,\forall k,\,\,\,\omega\geq 0,\,\,\,\boldsymbol{\Theta}_{k}\succeq\mathbf{0},\forall k\in [K+1]. \label{eq_kkt_WDF_fea_mu}
\end{align}
\end{subequations} 

Before proceeding, we recall the definition $\mathbf{Q}_{k}=\mathbf{h}_{k}\mathbf{h}_{k}^H$ for $k \in [K]$, which will be utilized frequently. The multipliers $\{\mu_k\}_{k=1}^K$ may not all be strictly positive. Let $T$ be the number of strictly positive multipliers, where $0\leq T\leq K$. Without loss of generality, we assume that $\mu_1,\mu_2,\ldots,\mu_T>0$ and $\mu_{T+1}=\cdots=\mu_K=0$. 
 
 Right-multiplying \eqref{eq_kkt_w1_gra_mu} by $\mathbf{W}_{k}$ for $k \in [K]$ and \eqref{eq_kkt_wA_gra_mu} by $\mathbf{W}_{K+1}$, then summing the resulting $K+1$ equalities and applying the complementary slackness condition in \eqref{eq_kkt_slackness}, we obtain
$ \omega\mathbf{R}_{W}-\mathbf{R}_{W}^{-1}=\sum_{k=1}^K\mu_k\left(\rho_k\mathbf{Q}_k\mathbf{W}_k-\mathbf{Q}_k\mathbf{R}_{W}\right)$.
Denote the eigenvalue decomposition of $\mathbf{R}_{W}$ by $\mathbf{U}\boldsymbol{\Lambda}\mathbf{U}^H$. Then, we have 
\begin{equation}\label{eq_kkt_sum_mu_1}
 \mathbf{U}\left(\omega\boldsymbol{\Lambda}-\boldsymbol{\Lambda}^{-1}\right)\mathbf{U}^H{=}\sum_{k=1}^T\mu_k\left(\rho_k\mathbf{Q}_k\mathbf{W}_k-\mathbf{Q}_k\mathbf{R}_{W}\right)\!{,}
\end{equation}
where we have used the fact that $\mu_{T+1}=\cdots=\mu_K=0$. 
Let $\mathbf{H}_T=[\mathbf{h}_1,\mathbf{h}_2,\ldots,\mathbf{h}_T]$.
 Note that each column of the matrix on the right-hand side of \eqref{eq_kkt_sum_mu_1} lies in the range space of $\mathbf{H}_T$, and the rank of this matrix is at most $r=\text{rank}(\mathbf{H}_T)$. This implies that at least $N_t-r$ eigenvalues of $\boldsymbol{\Lambda}$ equal to $1/\sqrt{\omega}$. Therefore, we can express $\mathbf{R}_{W}$ as
 \begin{equation}\label{eq_evd_Rx}
 \begin{aligned}
\mathbf{R}_{W}&=\underbrace{\left(\mathbf{U}_T,\mathbf{U}_A\right)
\left(\begin{array}{cc}\boldsymbol{\Lambda}_T &\mathbf{0} \\ \mathbf{0} &\frac{1}{\sqrt{\omega}}\mathbf{I}_{N_t-r} \end{array}\right)
\left(\begin{array}{c}\mathbf{U}_T^H \\ \mathbf{U}_A^H \end{array}\right)}_{\mathbf{U}\boldsymbol{\Lambda}\mathbf{U}^H}\\
&=\mathbf{U}_T\boldsymbol{\Lambda}_T\mathbf{U}_T^H+\frac{1}{\sqrt{\omega}}\mathbf{U}_A\mathbf{U}_A^H,  
 \end{aligned}
 \end{equation}
 where $\boldsymbol{\Lambda}_T\in \mathbb{C}^{r\times r}
$ is a diagonal matrix, and the columns of $\mathbf{U}_T\in\mathbb{C}^{N_t\times r}$ and $\mathbf{U}_A\in\mathbb{C}^{N_t\times (N_t-r)}$ are orthonormal eigenvectors of $\mathbf{R}_{W}$ corresponding to the eigenvalues in $\boldsymbol{\Lambda}_T$ and the eigenvalue $1/\sqrt{\omega}$, respectively. Moreover, we have $\mathcal{R}(\mathbf{U}_T)=\mathcal{R}(\mathbf{H}_T)$ and $\mathcal{R}(\mathbf{U}_A)=\mathcal{N}(\mathbf{H}_T^H)$ according to the above analysis.

We now analyze the optimal solution structure in three steps.

\emph{Step 1:} We show that $\{\mathbf{W}_k\}_{k=T+1}^K$ and $\mathbf{W}_{K+1}$ are orthogonal to $\mathbf{H}_T$.

Define  $\mathbf{B}=-\mathbf{R}_{W}^{-2}+\omega\mathbf{I}_{N_t}+\sum_{j=1}^T\mu_j\mathbf{Q}_{j}$. 
Right-multiplying \eqref{eq_kkt_w1_gra_mu} by $\mathbf{W}_{k}$ for $k \in \{ T+1, \ldots, K \}$, we obtain
\begin{equation}\label{eq_KKT_bwk}
    \mathbf{B}\mathbf{W}_{k}=\mathbf{0},~ k \in \{ T+1, \ldots, K \}.
\end{equation}
Right-multiplying \eqref{eq_kkt_wA_gra_mu} by $\mathbf{W}_{K+1}$, we have
\begin{equation}\label{eq_KKT_bwk1}
    \mathbf{B}\mathbf{W}_{K+1}=\mathbf{0}.
\end{equation}
According to \eqref{eq_kkt_w1_gra_mu}, we have
\begin{equation}\notag
    \mathbf{B}=\mu_k\rho_k\mathbf{Q}_k+\boldsymbol{\Theta}_k,\forall k\in [T].
\end{equation} 
Since $\mu_k>0$ for all $k\in[T]$ and $\boldsymbol{\Theta}_k\succeq\mathbf{0}$, we have $\mathcal{R}(\mathbf{h}_k)\subseteq \mathcal{R}(\mathbf{B})$ for all $k\in [T]$. This immediately implies that $\mathcal{R}(\mathbf{H}_T)\subseteq \mathcal{R}(\mathbf{B})$. It follows that \eqref{eq_KKT_bwk} and \eqref{eq_KKT_bwk1} give $\mathbf{H}_T^H\mathbf{W}_k=\mathbf{0}$ for all $k \in \{ T+1, \ldots, K \}$ and $\mathbf{H}_T^H\mathbf{W}_{K+1}=\mathbf{0}$, which completes the proof of the first step.

\emph{Step 2:} We construct an optimal solution such that $\mathcal{R}(\mathbf{W}_k)\subseteq \mathcal{R}(\mathbf{H}_T)$ for $k \in [T]$.

Given an optimal solution $\{\mathbf{W}_{k}\}_{k=1}^{K+1}$ to problem \eqref{eq_extended_mu_isac_problem_sdr}, which satisfies \eqref{eq_kkt_conditions}, define $\widebar{\mathbf{W}}_{k}= \mathbf{U}_T\mathbf{U}_T^H\mathbf{W}_{k}\mathbf{U}_T\mathbf{U}_T^H$ for $k \in [T]$. Observe that
\begin{equation}\label{eq_wt}
\begin{aligned}
     \sum_{k=1}^T \widebar{\mathbf{W}}_{k} &= \mathbf{U}_T\mathbf{U}_T^H\left(\sum_{k=1}^T\mathbf{W}_{k}\right)\mathbf{U}_T\mathbf{U}_T^H\\
&=\mathbf{U}_T\mathbf{U}_T^H\left(\sum_{k=1}^T\mathbf{W}_{k}+\sum_{k=T+1}^{K+1}\mathbf{W}_{k}\right)\mathbf{U}_T\mathbf{U}_T^H\\
&=\mathbf{U}_T\boldsymbol{\Lambda}_T\mathbf{U}_T^H,
\end{aligned}
  \end{equation}
where the second equality uses the conclusion in Step 1 that $\mathbf{U}_T^H\mathbf{W}_k=\mathbf{0}$ for $k \in \{ T+1, \ldots, K+1 \}$, and the third equality follows from \eqref{eq_evd_Rx} by noticing that ${\mathbf U}_T^H {\mathbf U}_A = {\mathbf 0}$. In addition, define $\widebar{\mathbf{W}}_{k}= \mathbf{W}_{k}$ for $k \in \{ T+1, \ldots, K \}$ and $\widebar{\mathbf{W}}_{K+1}= \mathbf{R}_{W}-\sum_{k=1}^K\widebar{\mathbf{W}}_{k}$. Since $\mathcal{R}( {\mathbf U}_T ) = \mathcal{R}( {\mathbf H}_T )$, we see that $\{\widebar{\mathbf{W}}_{k}\}_{k=1}^{K+1}$ is also an optimal solution to problem \eqref{eq_extended_mu_isac_problem_sdr}. With this construction, we have $\mathcal{R}(\widebar{\mathbf{W}}_{k})\subseteq \mathcal{R}\left(\mathbf{H}_T\right)$ for $k\in [T]$ and
\begin{equation}\label{eq_orthonormal_divide}
    \sum_{k=T+1}^{K+1}\widebar{\mathbf{W}}_{k}=\frac{1}{\sqrt{\omega}}\mathbf{U}_A\mathbf{U}_A^H
\end{equation}
according to \eqref{eq_evd_Rx} and \eqref{eq_wt}. The proof of Step 2 is completed.

\emph{Step 3:} Note that the columns of $\{\widebar{\mathbf{W}}_k\}_{k=T+1}^K$ may not lie in the range space of ${\mathbf H}$, and $\widebar{\mathbf{W}}_{K+1}$ may not be orthogonal to $\mathbf{H}$. However, we can construct from $\{ \bar{\mathbf W}_k \}_{k=1}^{K+1}$ an optimal solution to problem \eqref{eq_extended_mu_isac_problem_sdr} satisfying properties (a), (b), and (c) as follows. 

Recall that ${\mathbf U}_A$ is an orthonormal basis of the eigenspace of ${\mathbf R}_W$ corresponding to the eigenvalue $1/\sqrt{\omega}$. Since we have the orthogonal subspace decomposition  $\mathcal{R}( {\mathbf U}_A ) = ( \mathcal{R}( {\mathbf H} )\cap \mathcal{N}( {\mathbf H}_T^H )) \oplus \mathcal{N}( {\mathbf H}^H)$, we may assume without loss that ${\mathbf U}_A = \begin{pmatrix} {\mathbf U}_B & {\mathbf U}_C \end{pmatrix}$, where $\mathbf{U}_B$ and $\mathbf{U}_C$ are orthonormal bases of $\mathcal{R}( {\mathbf H} )\cap \mathcal{N}( {\mathbf H}_T^H )$ and $\mathcal{N}( {\mathbf H}^H)$, respectively. It then follows from \eqref{eq_orthonormal_divide} that
\begin{equation}\label{eq_W_bar_UbUc}
    \sum_{k=T+1}^{K+1}\widebar{\mathbf{W}}_{k}=\frac{1}{\sqrt{\omega}}\mathbf{U}_B\mathbf{U}_B^H+\frac{1}{\sqrt{\omega}}\mathbf{U}_C\mathbf{U}_C^H.
\end{equation}

\noindent Now, we construct $ \widehat{\mathbf{W}}_{k} =\mathbf{U}_B\mathbf{U}_B^H\widebar{\mathbf{W}}_{k}\mathbf{U}_B\mathbf{U}_B^H$ for $k \in \{ T+1, \ldots, K \}$. Then, we have
\begin{equation}\label{eq_W_hat_construct}
    \begin{aligned}
       \sum_{k=T+1}^K \widehat{\mathbf{W}}_{k} =\,\,&\mathbf{U}_B\mathbf{U}_B^H\left(\sum_{k=T+1}^K \widebar{\mathbf{W}}_{k}\right)\mathbf{U}_B\mathbf{U}_B^H\\
       \preceq \,\,&\mathbf{U}_B\mathbf{U}_B^H\left(\frac{1}{\sqrt{\omega}}\mathbf{U}_A\mathbf{U}_A^H\right)\mathbf{U}_B\mathbf{U}_B^H\\
        = & \,\, \frac{1}{\sqrt{\omega}} \mathbf{U}_B\mathbf{U}_B^H , 
    \end{aligned}
\end{equation}
where the second inequality is due to \eqref{eq_orthonormal_divide}, and the third equality uses the fact that $\mathcal{R}(\mathbf{U}_B)\subseteq \mathcal{R}(\mathbf{U}_A)$ and $\mathbf{U}_A\mathbf{U}_A^H$ is an orthogonal projection matrix. We further construct 
\begin{equation}\label{eq_proof_construct_W}
    \mathbf{W}^{\star}_{T+1} =\frac{1}{\sqrt{\omega}} \mathbf{U}_B\mathbf{U}_B^H-\sum_{k=T+2}^K \widehat{\mathbf{W}}_{k},
\end{equation}
$\mathbf{W}^{\star}_{k}=\widehat{\mathbf{W}}_{k}$ for $k\in \{T+2,\ldots,K\}$, $\mathbf{W}^{\star}_{k} =\widebar{\mathbf{W}}_{k}$ for $k\in [T]$, and $\mathbf{W}^{\star}_{K+1}=\frac{1}{\sqrt{\omega}} \mathbf{U}_C\mathbf{U}_C^H$. 

We show that $\{\mathbf{W}^{\star}_{k}\}_{k=1}^{K+1}$ is an optimal
solution to problem \eqref{eq_extended_mu_isac_problem_sdr}. First, from the above construction and the equality \eqref{eq_W_bar_UbUc}, we have
\begin{equation}\notag
    \sum_{k=1}^{K+1} \mathbf{W}_k^\star = \sum_{k=1}^{K+1} \widebar{\mathbf{W}}_k.
\end{equation}
Thus, the objective value achieved by $\{\mathbf{W}^{\star}_{k}\}_{k=1}^{K+1}$ is the same as that of $\{\widebar{\mathbf{W}}_{k}\}_{k=1}^{K+1}$, and the power constraint is satisfied since $\{\widebar{\mathbf{W}}_{k}\}_{k=1}^{K+1}$ is an optimal solution to problem \eqref{eq_extended_mu_isac_problem_sdr}. Note that the SINR of user $k$ in problem \eqref{eq_extended_mu_isac_problem_sdr} can be expressed as 
\begin{equation}\label{eq_SINR_k}
   \text{SINR}_k= \rho_k\operatorname{tr}\left(\mathbf{Q}_{k}\mathbf{W}_{k}\right)-\operatorname{tr}\left(\mathbf{Q}_{k}\left(\sum_{j=1}^{K+1}\mathbf{W}_{j}\right)\right).
\end{equation}
The second term in \eqref{eq_SINR_k} gives the same value when evaluated at $\{\mathbf{W}^{\star}_{k}\}_{k=1}^{K+1}$ and $\{\widebar{\mathbf{W}}_{k}\}_{k=1}^{K+1}$. Therefore, we only need to compare the first term. For $k \in \{T+1,...,K\}$, we have
\begin{equation}\label{eq_stpe3_1}
    \begin{aligned}
        \operatorname{tr}\left(\mathbf{Q}_{k}\widehat{\mathbf{W}}_{k}\right)
        &=\operatorname{tr}\left(\mathbf{Q}_{k}\left(\mathbf{U}_B\mathbf{U}_B^H\widebar{\mathbf{W}}_{k}\mathbf{U}_B\mathbf{U}_B^H\right)\right)\\
        &=\operatorname{tr}\left(\mathbf{U}_B\mathbf{U}_B^H\mathbf{Q}_{k}\mathbf{U}_B\mathbf{U}_B^H\widebar{\mathbf{W}}_{k}\right)
    \end{aligned}
\end{equation}
according to the definition of $\widehat{\mathbf{W}}_{k}$. Moreover, we have
\begin{equation}\label{eq_stpe3_2}
    \begin{aligned}
        &\operatorname{tr}\left(\mathbf{Q}_{k}\widebar{\mathbf{W}}_{k}\right)\\
        &=\operatorname{tr}\left(\mathbf{h}_{k}\mathbf{h}_{k}^H\widebar{\mathbf{W}}_{k}\right)\\
         &=\operatorname{tr}\left(\left(\mathbf{U}_B\mathbf{U}_B^H\mathbf{h}_{k}+\mathbf{U}_T\mathbf{U}_T^H\mathbf{h}_{k}\right)\right.\\
         &\quad\quad\quad\quad\quad\times\left.\left(\mathbf{U}_B\mathbf{U}_B^H\mathbf{h}_{k}+\mathbf{U}_T\mathbf{U}_T^H\mathbf{h}_{k}\right)^H\widebar{\mathbf{W}}_{k}\right)\\
         &=\operatorname{tr}\left(\mathbf{U}_B\mathbf{U}_B^H\mathbf{Q}_{k}\mathbf{U}_B\mathbf{U}_B^H\widebar{\mathbf{W}}_{k}\right),
    \end{aligned}
\end{equation}
where the second equality holds since the columns of $\mathbf{U}_B$ and $\mathbf{U}_T$ together form an orthonormal basis of $\mathcal{R}( {\mathbf H} )$ and $\mathbf{h}_k\in \mathcal{R}( {\mathbf H} )$, which leads to $\mathbf{U}_B\mathbf{U}_B^H\mathbf{h}_{k}+\mathbf{U}_T\mathbf{U}_T^H\mathbf{h}_{k}=\mathbf{h}_{k}$; the last equality uses the result from Step 1 that  $\mathbf{H}_T^H\widebar{\mathbf{W}}_{k}=\mathbf{U}_T^H\widebar{\mathbf{W}}_{k}={\mathbf 0}$ for $k \in \{T+1,...,K\}$.
It follows from \eqref{eq_stpe3_1} and \eqref{eq_stpe3_2} that the SINR of user $k \in \{T+1,...,K\}$ remains unchanged when replacing $\{\widebar{\mathbf{W}}_k\}_{k=1}^{K+1}$ with $\{\widehat{\mathbf{W}}_{k}\}_{k=1}^{K+1}$. Since $\mathbf{W}^{\star}_{k}=\widehat{\mathbf{W}}_{k}$ for $k\in \{T+2,\ldots,K\}$, the SINR constraints for these users are satisfied by $\{\mathbf{W}^{\star}_{k}\}_{k=1}^{K+1}$. Moreover, combining the definition of $\mathbf{W}_{T+1}^{\star}$ in \eqref{eq_proof_construct_W} with the relation in \eqref{eq_W_hat_construct} yields $\mathbf{W}_{T+1}^{\star} \succeq \widebar{\mathbf{W}}_{T+1}$. This implies that the SINR of user $T+1$ does not decrease when replacing $\{ \widehat{\mathbf W}_k \}_{k=1}^{K+1}$ with $\{ {\mathbf W}_k^\star \}_{k=1}^{K+1}$, so the corresponding constraint is satisfied by $\{ {\mathbf W}_k^\star \}_{k=1}^{K+1}$. For user $k \in [T]$, the SINR constraint is clearly maintained because $\mathbf{W}^{\star}_{k} =\widebar{\mathbf{W}}_{k}$ for $k\in [T]$. 

In summary, we have constructed an optimal solution to problem \eqref{eq_extended_mu_isac_problem_sdr} satisfying $\mathcal{R}(\mathbf{W}^{\star}_{k})\subseteq \mathcal{R}(\mathbf{H})$ for $ k\in [K]$, $\mathbf{W}_{K+1}^{\star}\subseteq \mathcal{N}(\mathbf{H}^H)$, and $\mathbf{W}^{\star}_{K+1}=\frac{1}{\sqrt{\omega}} \mathbf{U}_C\mathbf{U}_C^H$. The value of $\theta= $ {\footnotesize \scalebox{1.1}{$\frac{1}{\sqrt{\omega}}$}} can be obtained by noting that $\operatorname{tr}(\mathbf{R}^\star_{ {\mathbf W} }) = \sum_{k=1}^K \operatorname{tr}( {\mathbf W}_k^\star ) + \operatorname{tr}( {\mathbf W}_{K+1}^\star ) = P_T$ by \eqref{eq_kkt_comp_power}. This completes the proof of Theorem \ref{theorem_low_dimen}.

\bibliographystyle{IEEEtran}
\bibliography{to_bibitem}

\newpage

\twocolumn[%
    \begin{center}
    \Huge Supplementary Material
    \end{center}\vspace{1cm}
]

\section{Proof of Theorem \ref{thm_extract_rank1}}\label{app_thm_extract_rank1}

Denote the optimal solution to problem \eqref{eq_extended_mu_isac_problem_wc_lds_sdr_brief} by $\{\mathbf{X}^{\star}_{k}\}_{k=1}^K$. According to Theorem \ref{theorem_low_dimen} and Corollary \ref{coro_reduced}, an optimal solution to problem \eqref{eq_extended_mu_isac_problem_sdr} is given by \begin{equation}\notag
  \widetilde{\mathbf{W}}_{k}=  \tilde{\mathbf{U}}\mathbf{X}^{\star}_{k}\tilde{\mathbf{U}}^H,\forall k
\end{equation}
and 
\begin{equation}\notag
     \widetilde{\mathbf{W}}^{K+1}=\frac{P_T-\sum_{k=1}^K\operatorname{tr}(\mathbf{X}^{\star}_{k})}{N_t-K}\mathbf{U}_C\mathbf{U}_C^H.
\end{equation}
Then, a method analogous to \cite[Theorem 1]{liu2020joint} can be employed to extract a rank-one solution from $\{\mathbf{W}^{\star}_{k}\}_{k=1}^K$. Through direct algebraic manipulation, we obtain \eqref{thm_extract_rank_1} and \eqref{thm_extract_rank_2}. 

\section{Proof of Theorem \ref{thm_equality}}\label{app_thm_equality}
\subsection{Proof of Theorem \ref{thm_equality}(a)}
 To characterize the degenerate case, we first present the following key property.
 \begin{lemma}\label{lemma_multiplier_vanish}
     If problem \eqref{eq_extended_mu_isac_problem_wc_lds_sdr_brief} is degenerate, then the optimal Lagrangian multipliers $\{ \mu_k \}_{k=1}^K$ to \eqref{eq_kkt_conditions} satisfy $\mu_k=0$ for $k\in[K]$.
 \end{lemma}
 \begin{proof}
 We prove the lemma by contradiction. Consider a solution $( {\mathbf W}, {\boldsymbol \mu}, \omega, {\mathbf \Theta} )$ to the KKT conditions in \eqref{eq_kkt_conditions}, and suppose to the contrary that at least one of the Lagrangian multipliers in $\{ \mu_k \}_{k=1}^K$ is strictly positive. Without loss of generality, assume that $\mu_k>0$ for $k\in [T]$, where $1\leq T\leq K$, and $\mu_j=0$ for $j\in \{T+1,\ldots,K\}$. By the complementary slackness condition in \eqref{eq_kkt_comp}, the SINR constraint for $k\in [T]$ is tight. For each $j\in \{T+1,\ldots,K\}$, if the corresponding SINR constraint is not tight, then we can scale down the beamforming matrix via $\bar{\mathbf{W}}_j=\alpha_j\mathbf{W}_j$ with $\alpha_j\in(0,1) $, such that the $j$-th constraint becomes tight. To ensure that the overall beamforming matrix ${\mathbf R}_W$ is unchanged, we allocate the excess $(1-\alpha_j){\mathbf W}_j$ to the beamforming matrix $\mathbf{W}_1$ by defining $\bar{\mathbf{W}}_1=\mathbf{W}_1+(1-\alpha_j)\mathbf{W}_j\succeq {\mathbf 0}$ . 
 Crucially, as shown in Appendix \ref{app_theorem_low_dimen}, we have $\mathbf{H}_T^H \mathbf{W}_j = \mathbf{0}$ for $j\in \{T+1,\ldots,K\}$. Thus, this adjustment does not affect the tightness of the first $T$ SINR constraints. Repeating this process for all non-tight SINR constraints, we eventually obtain an optimal solution to problem \eqref{eq_extended_mu_isac_problem_sdr} for which all SINR constraints are tight. The resulting solution still possesses the low-dimensional structures described in Theorem \ref{theorem_low_dimen}. This contradicts the assumption that problem \eqref{eq_extended_mu_isac_problem_wc_lds_sdr_brief} is degenerate.
 \end{proof}

According to Lemma \ref{lemma_multiplier_vanish}, the degenerate case occurs only when all optimal Lagrangian multipliers to \eqref{eq_kkt_conditions} vanish. It then follows from \eqref{eq_evd_Rx} (with $T=r=0$) that the optimal solution $\{ {\mathbf W}_k \}_{k=1}^{K+1}$ to problem \eqref{eq_extended_mu_isac_problem_sdr} satisfies $\mathbf{R}_{W}=\frac{P_T}{N_t}\mathbf{I}_{N_t}$. By further exploiting the low-dimensional structures in Theorem \ref{theorem_low_dimen}, we see that an optimal solution $\{ {\mathbf X}_k \}_{k=1}^K$ to problem \eqref{eq_extended_mu_isac_problem_wc_lds_sdr_brief} satisfies
\begin{subequations}\label{eq_degenerate}
\begin{align}
  &\sum_{k=1}^{K}\mathbf{X}_k=\frac{P_T}{N_t}\mathbf{I}_{K},~\mathbf{X}_k\succeq\mathbf{0},\forall k,\label{eq_degenerate1}\\
    &\rho_k\tr\left(\tilde{\mathbf{Q}}_k\mathbf{X}_k\right)\geq \frac{P_T}{N_t}\|\mathbf{h}_k\|^2+\sigma_C^2,\forall k.\label{eq_degenerate2}
    \end{align}
\end{subequations}

 To establish the desired result in Theorem \ref{thm_equality}(a), it suffices to prove that the occurrence of the degenerate case in problem \eqref{eq_extended_mu_isac_problem_wc_lds_sdr_brief}, i.e., at least one of the constraints in \eqref{eq_degenerate2} is not tight, will lead to condition \eqref{eq_inequality_condition}. We prove this by contradiction. Without loss of generality, let us assume that the contrary holds for index $l=1$, i.e.,
 \begin{equation}\label{eq_inequality_condition_con}
     \|\mathbf{h}_1\|^2 \sum_{k=1}^K \frac{P_T\|\mathbf{h}_k\|^2+  \sigma_C^2 N_t}{\rho_k\tr\left(\mathbf{Q}_k\mathbf{Q}_1\right)}\geq P_T. 
    \end{equation}
 Under this assumption, we construct an optimal solution to problem \eqref{eq_extended_mu_isac_problem_wc_lds_sdr_brief} for which all SINR constraints in \eqref{eq_degenerate2} are tight. The construction consists of three steps.

\emph{Step 1:}  We construct a solution that satisfies \eqref{eq_degenerate1}. For this solution, the first SINR constraint will be satisfied (either as equality or not), while all remaining SINR constraints are tight.

Specifically, starting from an optimal solution $\{\mathbf{X}_k\}_{k=1}^K$ to problem \eqref{eq_extended_mu_isac_problem_wc_lds_sdr_brief} that satisfies \eqref{eq_degenerate}, we construct a new solution $\{\check{\mathbf{X}}_k\}_{k=1}^K$ such that 
 \begin{equation}\notag
\rho_1\tr\left(\tilde{\mathbf{Q}}_1\check{\mathbf{X}}_1\right)\geq \frac{P_T}{N_t}\|\mathbf{h}_1\|^2+\sigma_C^2
\end{equation}
and
\begin{equation}\notag
    \rho_k\tr\left(\tilde{\mathbf{Q}}_k\check{\mathbf{X}}_k\right)= \frac{P_T}{N_t}\|\mathbf{h}_k\|^2+\sigma_C^2,~ \forall k\in \{2,3,\ldots,K\}.
\end{equation}
This is achieved by scaling down $\mathbf{X}_k$ so that the inequalities in \eqref{eq_degenerate2} are tight for $k\in \{2,3,\ldots,K\}$ and allocating the excesses to $\mathbf{X}_1$ just as in the proof of Lemma \ref{lemma_multiplier_vanish}. As a result, the equality $\sum_{k=1}^K\check{\mathbf{X}}_k=\sum_{k=1}^K\mathbf{X}_k$ still holds, and thus so does \eqref{eq_degenerate1}.

\emph{Step 2:} We construct a solution $\{ \hat{\mathbf X}_k \}_{k=1}^K$ with the properties that (i) for $k \in \{2,3,\ldots,K\}$, the component of $\hat{\mathbf X}_k$ lying in the range space of $\tilde{\mathbf Q}_1$ is maximized, (ii) the constraints in \eqref{eq_degenerate1} are satisfied, (iii) the SINR constraints in  \eqref{eq_degenerate2} are tight for $k \in \{2,3,\ldots,K\}$, and (iv) the first SINR constraint is violated. The reason for property (iv) will become clear in Step 3.



Specifically, we define $\hat{\mathbf{X}}_k= a_k^* \tilde{\mathbf{Q}}_1+\mathbf{F}_k^*$ for $k \in \{2,3,\ldots,K\}$, where the values $\{a_k^*\}_{k=2}^K$ and matrices $\{\mathbf{F}_k^*\}_{k=2}^K$ are defined as the optimal solution to the following problem:
    \begin{equation}\label{eq_auxi_opt}
\begin{aligned}
 \max_{\{a_k,\mathbf{F}_k\}_{k=2}^K} ~~&\sum_{k=2}^{K} a_k\\
 \text{ s.t. } ~~~~~&\rho_k\tr\left(\tilde{\mathbf{Q}}_k\left(a_k \tilde{\mathbf{Q}}_1+\mathbf{F}_k\right)\right)\\
 &~~~~~~= \frac{P_T}{N_t}\|\mathbf{h}_k\|^2+\sigma_C^2,~ \forall k\in \{2,3,\ldots,K\},\\
 &\sum_{k=2}^K\left(a_k \tilde{\mathbf{Q}}_1+\mathbf{F}_k\right)\preceq \frac{P_T}{N_t}\mathbf{I}_K,\\
 &a_k\geq 0,~\mathbf{F}_k\succeq\mathbf{0},~\forall k\in \{2,3,\ldots,K\},\\
 &\tilde{\mathbf{Q}}_1^H\mathbf{F}_k=\mathbf{0},~ \forall k\in \{2,3,\ldots,K\}.
\end{aligned}
\end{equation}
Here, the last constraint stipulates that ${\mathbf F}_k$ lies in the orthogonal complement of the range space of $\tilde{\mathbf Q}_1$ for $k \in \{2,3,\ldots,K\}$. By applying suitable orthogonal transformations if necessary, we see that the solution $\{ \check{\mathbf X}_k \}_{k=1}^K$ constructed in Step 1 is feasible for problem \eqref{eq_auxi_opt}. We then define $\hat{\mathbf{X}}_1= \frac{P_T}{N_t}\mathbf{I}_K-\sum_{k=2}^K\hat{\mathbf{X}}_k$. We claim that the resulting solution $\{\hat{\mathbf{X}}_k\}_{k=1}^K$ satisfies
\begin{equation}\label{eq_auxi_opt2}
    \rho_1\tr\left(\tilde{\mathbf{Q}}_1\hat{\mathbf{X}}_1\right)< \frac{P_T}{N_t}\|\mathbf{h}_1\|^2+\sigma_C^2
\end{equation}
and
\begin{equation}\label{eq_auxi_opt3}
    \rho_k\tr\left(\tilde{\mathbf{Q}}_k\hat{\mathbf{X}}_k\right)= \frac{P_T}{N_t}\|\mathbf{h}_k\|^2+\sigma_C^2,~ \forall k\in \{2,3,\ldots,K\}
\end{equation}
under assumption \eqref{eq_inequality_condition_con}. The property \eqref{eq_auxi_opt3} follows directly from the constraints in problem \eqref{eq_auxi_opt}. Now, note that the second and fourth constraints in problem \eqref{eq_auxi_opt} give $\|\mathbf{h}_1\|^2\sum_{k=2}^{K} a_k^* \leq P_T/N_t$. To verify \eqref{eq_auxi_opt2}, we consider the following two cases: 
\begin{itemize}
    \item Case I: If $\|\mathbf{h}_1\|^2\sum_{k=2}^{K} a_k^* =P_T/N_t$, then by definition of $\{ \hat{\mathbf X}_k \}_{k=1}^K$, we have $\tr(\tilde{\mathbf{Q}}_1\hat{\mathbf{X}}_1)=0$, which implies that \eqref{eq_auxi_opt2} is satisfied. 
    \item Case II: If $\|\mathbf{h}_1\|^2\sum_{k=2}^{K} a_k^*<P_T/N_t$, then we claim that $\tilde{\mathbf{Q}}_1\tilde{\mathbf{Q}}_k\neq \mathbf{0}$ for $k\in \{2,3,\ldots,K\}$. To prove this, suppose to the contrary that $\tilde{\mathbf{Q}}_1\tilde{\mathbf{Q}}_l= \mathbf{0}$ for some $l\in \{2,3,\ldots,K\}$. Then, we can increase $a_l^*$ until $\|\mathbf{h}_1\|^2\sum_{k=2}^{K} a_k^* =P_T/N_t$ is satisfied while keeping all other variables in problem \eqref{eq_auxi_opt} unchanged. Moreover, all other constraints in problem \eqref{eq_auxi_opt} are still satisfied. This contradicts the optimality of $ a_l^*$. Hence, we must have $\tilde{\mathbf{Q}}_1\tilde{\mathbf{Q}}_k\neq \mathbf{0}$ for $k\in \{2,3,\ldots,K\}$. Next, we claim that $\mathbf{F}_k^*=\mathbf{0}$ for $k\in \{2,3,\ldots,K\}$. Indeed, suppose to the contrary that $\mathbf{F}_l^*\neq \mathbf{0}$ for some $l\in \{2,3,\ldots,K\}$. Then, we can set ${\mathbf F}_l^* = {\mathbf 0}$ and increase $a_l^*$ while preserving the tightness of the SINR constraint of user $l$. Consequently, we have $\hat{\mathbf{X}}_k=a_k^*\tilde{\mathbf{Q}}_k$ for $k\in \{2,3,\ldots,K\}$. Given that \eqref{eq_auxi_opt3} holds and using the fact that $\tr\left(\mathbf{Q}_k\mathbf{Q}_1\right)=\tr(\tilde{\mathbf{Q}}_k\tilde{\mathbf{Q}}_1)$ from the definitions of $\{\tilde{\mathbf{Q}}_k\}_{k=1}^K$, we obtain $a_k^*=\frac{P_T\|\mathbf{h}_k\|^2+  \sigma_C^2 N_t}{\rho_k N_t\tr\left(\mathbf{Q}_k\mathbf{Q}_1\right)}$ for $k\in \{2,3,\ldots,K\}$. Moreover, we have
    \begin{equation}\notag
    \begin{aligned}
        &\rho_1\tr\left(\tilde{\mathbf{Q}}_1\hat{\mathbf{X}}_1\right)\\
        &=\rho_1\frac{ P_T}{N_t}\|\mathbf{h}_1\|^2-\rho_1\tr\left(\tilde{\mathbf{Q}}_1\sum_{k=2}^K\hat{\mathbf{X}}_k\right)\\
        &=\rho_1\frac{ P_T}{N_t}\|\mathbf{h}_1\|^2-\rho_1\|\mathbf{h}_1\|^4\sum_{k=2}^K\frac{P_T\|\mathbf{h}_k\|^2+  \sigma_C^2 N_t}{\rho_k N_t\tr\left(\mathbf{Q}_k\mathbf{Q}_1\right)}\\
        &\leq \frac{P_T}{N_t}\|\mathbf{h}_1\|^2+\sigma_C^2.
        \end{aligned}
    \end{equation}
    Here, the first equality uses the definition of $\hat{\mathbf{X}}_1$; the second equality uses $\hat{\mathbf{X}}_k=a_k^*\tilde{\mathbf{Q}}_1$ for $k\in \{2,3,\ldots,K\}$ and the optimal value of $a_k^*$; the last inequality follows from the assumption \eqref{eq_inequality_condition_con}. Since problem \eqref{eq_extended_mu_isac_problem_wc_lds_sdr_brief} is assumed to be degenerate, the inequality above must be strict. It follows that \eqref{eq_auxi_opt2} holds.
\end{itemize}

\emph{Step 3:} We construct $\{\tilde{\mathbf{X}}_k\}_{k=1}^K$ by taking a convex combination of the solutions constructed in Step 1 and Step 2, i.e., $\tilde{\mathbf{X}}_k= \theta\check{\mathbf{X}}_k+(1-\theta)\hat{\mathbf{X}}_k$ for $k\in [K]$, where $\theta\in [0,1]$ is chosen so that all the SINR constraints in \eqref{eq_degenerate2} hold as equality. Such a choice exists due to the continuity of the trace function. Moreover, the constraints in \eqref{eq_degenerate1} hold for all $\theta\in [0,1]$. This contradicts the assumption that problem \eqref{eq_extended_mu_isac_problem_wc_lds_sdr_brief} is degenerate and completes the proof of Theorem \ref{thm_equality}(a).

\subsection{Proof of Theorem \ref{thm_equality}(b)}
Given $l\in[K]$, define $\mathbf{W}_k=a_k\mathbf{Q}_l$ for $k\in[K]$, where $a_k =  \frac{P_T\|\mathbf{h}_k\|^2+  \sigma_C^2 N_t}{\rho_k N_t\tr\left(\mathbf{Q}_k\mathbf{Q}_l\right)}$.
If condition \eqref{eq_inequality_condition} is satisfied, then we have 
 \begin{equation}\notag
      \sum_{k=1}^K\mathbf{W}_k= \sum_{k=1}^K a_k\mathbf{Q}_l\preceq \frac{P_T}{N_t}\mathbf{I}. 
    \end{equation}
We then define 
\begin{equation}\notag
    \mathbf{W}_{K+1}=\frac{P_T}{N_t}\mathbf{I}-\sum_{k=1}^Ka_k\mathbf{Q}_l\succeq\mathbf{0}.
\end{equation}
Now, set $\omega=(\frac{N_T}{P_T})^2$, $\mu_k=0$ for $k\in [K]$, and $\boldsymbol{\Theta}_k=\mathbf{0}$ for $k\in [K+1]$. We claim that $\{\mathbf{W}_k\}_{k=1}^{K+1}$, $\omega$, $\{\mu_k\}_{k=1}^{K}$, and $\{\boldsymbol{\Theta}_k\}_{k=1}^{K+1}$ satisfy the KKT conditions \eqref{eq_kkt_conditions} associated with problem \eqref{eq_extended_mu_isac_problem_sdr}. It is straightforward to verify all the conditions in \eqref{eq_kkt_conditions} except perhaps \eqref{eq_kkt_w1_fea_W}. To verify \eqref{eq_kkt_w1_fea_W}, we fix $k \in [K]$ and compute
\begin{equation}\notag
\begin{aligned}
    &\rho_k\operatorname{tr}\left(\mathbf{Q}_{k}\mathbf{W}_{k}\right)-\sum_{j=1}^{K+1}\operatorname{tr}\left(\mathbf{Q}_{k}\mathbf{W}_{j}\right)\\
    =\,\,&a_k\rho_k\operatorname{tr}\left(\mathbf{Q}_{k}\mathbf{Q}_l\right)-\frac{P_T}{N_T}\|\mathbf{h}_k\|^2\\
    =\,\,&\sigma^2_C,
\end{aligned}
    \end{equation}
where the first equality uses the definitions of $\{\mathbf{W}_k\}_{k=1}^{K+1}$ and the second follows from the definitions of $\{ a_k \}_{k=1}^K$.
This completes the proof.

\section{Proof of Proposition \ref{pro_subproblem_X}}\label{app_pro_subproblem_X}
The $\mathbf{X}$-subproblem in \eqref{eq_X_subproblem} can be reformulated as 
\begin{equation}\notag
  \begin{aligned}
 \min _{\{\mathbf{X}_k\}_{k=1}^K} ~&\sum_{k=1}^K\|\mathbf{X}_k-\widetilde{\mathbf{X}}_k\|_F^2\\
\text{s.t.} ~ ~~& \sum_{k=1}^K\operatorname{tr}\left(\mathbf{X}_k\right)\leq P_T,~ \mathbf{X}_k\succeq \mathbf{0},\forall k.
    \end{aligned}  
\end{equation}
Utilizing the eigenvalue decomposition $\widetilde{\mathbf{X}}_k=\mathbf{U}_k\boldsymbol{\Sigma}_k\mathbf{U}_k^H$ and exploiting the unitary invariance of the Frobenius norm, we obtain
\begin{equation}\notag
    \|\mathbf{X}_k-\widetilde{\mathbf{X}}_k\|_F^2=\|\mathbf{U}_k^H\mathbf{X}_k\mathbf{U}_k-\boldsymbol{\Sigma}_k\|_F^2.
\end{equation}
It is obvious that the optimal $\mathbf{U}_k^H\mathbf{X}_k\mathbf{U}_k$ must be a diagonal matrix, which implies that the optimal $\mathbf{X}_k$ takes the form $\mathbf{U}_k\boldsymbol{\Lambda}_k\mathbf{U}_k^H$ for some diagonal matrix $\boldsymbol{\Lambda}_k$. Therefore, it suffices to solve the following problem:
\begin{equation}\label{eq_subproblem_XX}
  \begin{aligned}
 \min _{\{\boldsymbol{\Lambda}_k\}_{k=1}^K} ~& \sum_{k=1}^K\sum_{l=1}^K\left( [\boldsymbol{\Lambda}_k]_{l,l}- [\boldsymbol{\Sigma}_k]_{l,l}\right)^2\\
\text{s.t.} ~~~  &\sum_{k=1}^K\sum_{l=1}^K[\boldsymbol{\Lambda}_k]_{l,l}\leq P_T, ~[\boldsymbol{\Lambda}_k]_{l,l}\geq 0,\forall k,l.
    \end{aligned}  
\end{equation} 

\noindent The KKT conditions associated with problem \eqref{eq_subproblem_XX} are given by
\begin{subequations}\label{eq_subproblem_kkt}
\begin{align}
    &2\left([\boldsymbol{\Lambda}_k]_{l,l}- [\boldsymbol{\Sigma}_k]_{l,l}\right)-\mu_{k,l}+\gamma=0,\forall k,l,\label{eq_subproblem_kkt_1}\\
    &\sum_{k=1}^K\sum_{l=1}^K[\boldsymbol{\Lambda}_k]_{l,l}\leq P_T,~[\boldsymbol{\Lambda}_k]_{l,l}\geq 0,\forall k,l,\\
   & \gamma\geq 0,~\mu_{k,l}\geq 0,\forall k,l, \label{eq_subproblem_kkt_3}\\
   &\gamma\left(\sum_{k=1}^K\sum_{l=1}^K[\boldsymbol{\Lambda}_k]_{l,l}- P_T\right)=0, \label{eq_subproblem_kkt_4}\\
   &\mu_{k,l}[\boldsymbol{\Lambda}_k]_{l,l}=0,\forall k,l.
    \end{align} 
\end{subequations}
From \eqref{eq_subproblem_kkt_1} and \eqref{eq_subproblem_kkt_3}, we obtain
\begin{equation}\label{eq_lambda}
    [\boldsymbol{\Lambda}_k]_{l,l}=\operatorname{max}\left\{[\boldsymbol{\Sigma}_k]_{l,l}-\frac{\gamma}{2},0\right\},\forall k,l.
\end{equation}
We now consider two cases based on the value of $\gamma$. First, if $\gamma=0$, then $\boldsymbol{\Lambda}_k=\boldsymbol{\Sigma}_k$ for all $ k\in [K]$. This implies that $\{ {\mathbf \Sigma}_k \}_{k=1}^K$ is an optimal solution to problem \eqref{eq_subproblem_XX}, provided that $\sum_{k=1}^K\sum_{l=1}^K[\boldsymbol{\Sigma}_k]_{l,l}\leq P_T$ holds. 
Second, if $\gamma>0$, then the complementary slackness condition in \eqref{eq_subproblem_kkt_4} implies that $\sum_{k=1}^K\sum_{l=1}^K[\boldsymbol{\Lambda}_k]_{l,l}= P_T$.
This, together with \eqref{eq_lambda} and the fact that $\gamma \mapsto g(\gamma) = \sum_{k=1}^K \sum_{l=1}^K \max \left\{ [ {\mathbf \Sigma}_k ]_{l,l} - \frac{\gamma}{2}, 0 \right\}$ is a piecewise linear, decreasing function on $\gamma \ge 0$, implies that we can determine the value of $\gamma$ by solving $g(\gamma) = P_T$.
Combining these two cases, we see that the optimal multiplier $\gamma$ is the smallest nonnegative value satisfying $g\left(\gamma\right)\leq P_T$. This completes the proof.

\section{Proof of Proposition \ref{pro_subproblem_Y}}\label{app_pro_subproblem_Y}
Similar to the $\mathbf{X}$-subproblem, the optimal solution to the $\mathbf{Y}$-subproblem \eqref{eq_Y_subproblem} takes the form $\mathbf{U}\boldsymbol{\Lambda}\mathbf{U}^H$, where $\boldsymbol{\Lambda}$ is a diagonal matrix obtained by solving the following problem:
\begin{equation}\notag
  \begin{aligned}
 \min _{\boldsymbol{\Lambda}} ~&\sum_{k=1}^K\left(\frac{1}{[\boldsymbol{\Lambda}]_{k,k}}+\frac{1}{2\tau}\left([\boldsymbol{\Lambda}]_{k,k}-[\boldsymbol{\Sigma}]_{k,k}\right)^2\right)\\
\text{s.t.} ~  &[\boldsymbol{\Lambda}]_{k,k}> 0,\forall k.
    \end{aligned}  
\end{equation}
Note that this problem is separable. Thus, it suffices to consider the one-dimensional subproblem
\begin{equation}\notag
    \min _{[\boldsymbol{\Lambda}]_{k,k}} ~\frac{1}{[\boldsymbol{\Lambda}]_{k,k}}+\frac{1}{2\tau}\left([\boldsymbol{\Lambda}]_{k,k}-[\boldsymbol{\Sigma}]_{k,k}\right)^2.
\end{equation}
Setting the derivative of the objective function with respect to $[\boldsymbol{\Lambda}]_{k,k}$ to zero yields the first-order optimality condition
\begin{equation}\notag
   ( [\boldsymbol{\Lambda}]_{k,k})^3-[\boldsymbol{\Sigma}]_{k,k} ([\boldsymbol{\Lambda}]_{k,k})^2-\tau=0.
\end{equation}
Let $h: \mathbb{R} \rightarrow \mathbb{R}$ be the cubic polynomial given by $h(x)= x^3-[\boldsymbol{\Sigma}]_{k,k} x^2-\tau$. We find that $h(0)<0$, $h'(0)=0$, and $h(+\infty)=+\infty$. Therefore, the equation $h(x)=0$ has exactly one positive real root.
The proof is completed. 
\section{Proof of Proposition \ref{pro_subproblem_Z}}\label{app_pro_subproblem_Z}
Similar to the $\mathbf{X}$-subproblem, the optimal solution to the $\mathbf{Z}$-subproblem \eqref{eq_Z_subproblem} is given by $\mathbf{U}\boldsymbol{\Lambda}\mathbf{U}^H$, where $\boldsymbol{\Lambda}$ is a diagonal matrix obtained by solving the following problem:
\begin{equation}\label{eq_ZZ_subproblem}
  \begin{aligned}
 \min _{\boldsymbol{\Lambda}} ~&\frac{\left(N_t-K\right)^2}{P_T-\sum_{k=1}^K[\boldsymbol{\Lambda}]_{k,k}}+\frac{1}{2\tau}\sum_{k=1}^K\left([\boldsymbol{\Lambda}]_{k,k}-[\boldsymbol{\Sigma}]_{k,k}\right)^2\\
\text{s.t.} ~  &[\boldsymbol{\Lambda}]_{k,k}\geq 0,\forall k.
    \end{aligned}  
\end{equation}
The KKT conditions associated with problem \eqref{eq_ZZ_subproblem} give
\begin{equation}\label{eq_lambda11}
    [\boldsymbol{\Lambda}]_{k,k}=\operatorname{max}\left\{[\boldsymbol{\Sigma}]_{k,k}-\lambda,0\right\},\forall k,
\end{equation}
where
\begin{equation}\label{eq_lambda22}
    \lambda=\frac{\tau(N_t-K)^2}{\left(P_T-\sum_{k=1}^K[\boldsymbol{\Lambda}]_{k,k}\right)^2}.
\end{equation} 
Substituting \eqref{eq_lambda11} into \eqref{eq_lambda22} gives 
\begin{equation}\label{eq_lambda33}
    \lambda\left(P_T-\sum_{k=1}^K\operatorname{max}\left\{[\boldsymbol{\Sigma}]_{k,k}-\lambda,0\right\}\right)^2-\tau(N_t-K)^2=0.
\end{equation}
Observe that the left-hand side of \eqref{eq_lambda33} is increasing in $\lambda$ and is negative when $\lambda=0$. It follows that equation \eqref{eq_lambda33} has a positive real root, which can be found by, e.g., a simple bisection search.

\vfill

\end{document}